\documentclass[a4paper,12pt]{article}
\usepackage[utf8]{inputenc}
\usepackage[UKenglish]{isodate}
\cleanlookdateon
\usepackage{geometry}
\usepackage{tikz}
\usetikzlibrary{matrix}
\usepackage{amssymb}
\usepackage{amsmath}
\usepackage{amsthm}
\usepackage{newpxtext,newpxmath}
\usepackage{csquotes}
\usepackage{xcolor}
\usepackage{enumitem}
\usepackage{thm-restate}
\usepackage{esvect}
\usepackage{subcaption}
\usepackage{pgfplots}
\usepgfplotslibrary{groupplots}

\definecolor{color2}{HTML}{003f5c}
\definecolor{color3}{HTML}{7a5195}
\definecolor{color4}{HTML}{ef5675}
\definecolor{color5}{HTML}{ffa600}

\usepackage[authoryear]{natbib}
\usepackage{hyperref}
\hypersetup{colorlinks,
  citecolor=orange!70!gray}

\newcommand{\floor}[1]{\left \lfloor #1 \right \rfloor}
\newcommand{\ceil}[1]{\left \lceil #1 \right \rceil}
\newcommand{\abs}[1]{\left \vert #1 \right \vert}
\newcommand{\eps}{\varepsilon}
\newcommand{\mbf}[1]{\mathbf{#1}}
\newcommand{\N}{\mathbb{N}}
\newcommand{\Prb}{\mathbb{P}}
\newcommand{\Lnk}{\vv{L}(n, \mbf{k})}
\newcommand{\cG}{\mathcal{G}}
\newcommand{\ov}[1]{#1^c}

\newcommand{\zero}{0}
\newcommand{\one}{1}

\newcommand\blfootnote[1]{
  \begingroup
  \renewcommand\thefootnote{}\footnote{#1}
  \addtocounter{footnote}{-1}
  \endgroup
}

\newenvironment{myquote}%
  {\list{}{\leftmargin=0.15in\rightmargin=0.15in}\item[]}%
  {\endlist}

\theoremstyle{plain}
\newtheorem{theorem}{Theorem}
\newtheorem*{theorem*}{Theorem}
\newtheorem{corollary}{Corollary}
\newtheorem{lemma}{Lemma}
\newtheorem{proposition}{Proposition}
\newtheorem{claim}{Claim}

\theoremstyle{definition}
\newtheorem{definition}{Definition}
\theoremstyle{remark}
\newtheorem{remark}{Remark}

\title{Game Connectivity and Adaptive Dynamics}
\author{Tom Johnston\thanks{School of Mathematics, University of Bristol, Bristol, BS8 1UG, UK.}\:\,\thanks{Heilbronn Institute for Mathematical Research, Bristol, UK.}
    \quad Michael Savery\footnotemark[2]\:\,\thanks{Mathematical Institute, University of Oxford, Oxford, OX2 6GG, UK.}
    \quad Alex Scott\footnotemark[3]\:\,\thanks{Supported by EPSRC grant EP/X013642/1.}
    \quad Bassel Tarbush\thanks{Department of Economics, University of Oxford, Oxford, OX1 3UQ, UK.}}
\date{\today}

\advance\textwidth2\evensidemargin
\begin{document}

\maketitle

\begin{abstract}
    \noindent We analyse the typical structure of games in terms of the connectivity properties of their best-response graphs. Our central result shows that, among games that are `generic' (without indifferences) and that have a pure Nash equilibrium, all but a small fraction are \emph{connected}, meaning that every action profile that is not a pure Nash equilibrium can reach every pure Nash equilibrium via best-response paths. This has important implications for dynamics in games. In particular, we show that there are simple, uncoupled, adaptive dynamics for which period-by-period play converges almost surely to a pure Nash equilibrium in all but a small fraction of generic games that have one (which contrasts with the known fact that there is no such dynamic that leads almost surely to a pure Nash equilibrium in \emph{every} generic game that has one). We build on recent results in probabilistic combinatorics for our characterisation of game connectivity.
    \blfootnote{\emph{Email}: \textsf{tom.johnston@bristol.ac.uk, \{savery,scott\}@maths.ox.ac.uk, bassel.tarbush@economics.ox.ac.uk}}
    \blfootnote{\emph{Thanks}: Sergiu Hart, Sam Jindani, George Mailath, Suraj Malladi, Bary Pradelski, John Quah, Marco Scarsini, Ludvig Sinander, Alex Teytelboym, Leeat Yariv, Peyton Young.}
    \blfootnote{\emph{Acknowledgment}: The simulations were run on the computational facilities of the \href{http://www.bris.ac.uk/acrc/}{Advanced Computing Research Centre, University of Bristol}.}
\end{abstract}

\section{Introduction}
A fundamental question at the heart of the literature on learning in games and distributed systems is whether there are adaptive dynamics that are guaranteed to lead to a Nash equilibrium in every game. Examples of adaptive dynamics include better- and best-response dynamics, fictitious play \citep*{fudenberg1998theory}, adaptive play \citep*{young1993evolution}, regret matching \citep*{hart2000simple}, regret testing \citep*{foster2006regret}, trial-and-error learning \citep*{young2009learning}, and many more.\footnote{\citet*{hart2005adaptive} distinguishes between three types of dynamics in games: learning, evolutionary, and adaptive. Learning requires high levels of rationality (e.g.\ \citealp{kalai1993rational}) whereas players in evolutionary dynamics instead mechanically inherit traits (e.g.\ \citealp{weibull1997evolutionary,hofbauer1998evolutionary,sandholm2010population}). Adaptive agents fall somewhere in between: they use relatively little information and take actions that respond to their environment according to basic decision heuristics in a generally improving way.} Several important results have outlined the boundary between the possible and the impossible, i.e.\ between classes of adaptive dynamics that are guaranteed to lead to a Nash equilibrium in every game and classes that lack such a guarantee (e.g.\ see \citealp{young2007possible} for an overview). A particularly influential impossibility result due to \cite*{hart2003uncoupled,hart2006stochastic} establishes that there is no simple adaptive dynamic that is guaranteed to lead to a pure Nash equilibrium in every game that has one, where `simple' qualifies the amount of information that each player has access to.

In this paper, we show that learning pure Nash equilibria via simple adaptive dynamics is not as hopeless an endeavour as the impossibility result of \citeauthor*{hart2003uncoupled} might suggest: we look at the space of all (ordinal and generic) games that have at least one pure Nash equilibrium and we show that simple adaptive dynamics lead to a pure Nash equilibrium in all but a small fraction of such games.
To establish this result, we study the `connectivity' properties of games---a concept that we formalise below---and we show that all but a quantifiably small fraction of games have a connectivity property that is conducive to equilibrium convergence.

What we do can be seen as a `beyond the worst-case' analysis of learning pure Nash equilibria in games. If one interprets simple adaptive dynamics as algorithms whose inputs are games, then \citeauthor*{hart2003uncoupled} have established that simple adaptive dynamics perform very poorly on their worst-case inputs: there are games (inputs) on which such dynamics (algorithms) do not lead to a pure Nash equilibrium; moreover, such `worst-case' games exist for any number of players $n \geq 3$. But, while algorithms are often assessed in terms of their worst-case performance, in practice we are often interested in how they perform on `typical' problem instances, and this has given rise to `beyond the worst-case' analysis of algorithms (\citealp{roughgarden2019beyond}; \citeyear[Chapter 1]{roughgarden2021beyond}). The simplex algorithm, for example, performs poorly on worst-case inputs but tends to perform well on typical problem instances. In this spirit, we characterize the structure of `typical' games and, in contrast with \citeauthor*{hart2003uncoupled}, we establish that simple adaptive dynamics perform very well on typical inputs, i.e.\ are guaranteed to lead to a pure Nash equilibrium in `typical' games.

The central argument of our paper is outlined in points (i)-(iii) below.

\begin{figure}
\centering
\begin{tikzpicture}[scale=0.75]
\begin{scope}[scale=0.8,xshift=-80,yshift=70,every node/.append style={yslant=0,xslant=0.8},yslant=0,xslant=0.8]
\draw[xstep=2cm,ystep=1cm,color=gray] (0,0) grid (4,2);
\node at (1,1.5) {\footnotesize $\one,\one,\one$}; \node at (3,1.5) {\footnotesize$\zero,\zero,\one$};
\node at (1,0.5) {\footnotesize$\zero,\one,\zero$}; \node at (3,0.5) {\footnotesize $\one,\zero,\zero$};
\node[] at (-1,0.5) {$\mathbf{B}$};
\node[] at (-1,1.5) {$\mathbf{A}$};
\node[] at (3,2.5) {$\mathbf{B}$};
\node[] at (1,2.5) {$\mathbf{A}$};
\node[] at (5,1) {$\mathbf{A}$};
\draw[white] (-1.5,0.3) -- (-1.5,2) node [pos=0.5,above,rotate=90,yshift=0.2cm] {\color{black}Player 1};
\draw[white] (0,3) -- (3.2,3) node [pos=0.5,above] {\color{black}Player 2};
\end{scope}

\begin{scope}[scale=0.8,xshift=-80,yshift=0,every node/.append style={yslant=0,xslant=0.8},yslant=0,xslant=0.8]
\draw[xstep=2cm,ystep=1cm,color=gray] (0,0) grid (4,2);

\node at (1,1.5) {\footnotesize $\one,\zero,\zero$};  \node at (3,1.5) {\footnotesize $\zero,\one,\zero$};
\node at (1,0.5) {\footnotesize $\zero,\zero,\one$}; \node at (3,0.5) {\footnotesize $\one,\one,\one$};

\node[] at (5,1) {$\mathbf{B}$};
\end{scope}

\draw[scale=0.8,white] (3.8,1) -- (3.8,3.75) node [pos=0.5,above,rotate=270] {\color{black}Player 3};

\begin{scope}[xshift=150,yshift=-20,
scale=3.5,
roundnode/.style={rectangle, draw=white, fill=white,inner sep=0,outer sep=0}]
\foreach \x in {1,2}
\foreach \y in {1,2}
\foreach \z in {1,2}
{\node[] (\z\x\y) at (\x,\y,\z) {$\circ$};} 
\node[] () at (1,2,1) {$\bullet$};
\node[] () at (2,1,2) {$\bullet$};

\path[->] (211) edge [thick ]  (111);
\path[->] (121) edge [thick ]  (221);

\path[->] (111) edge [thick ]  (121);
\path[->] (211) edge [thick ]  (221);

\path[->] (212) edge [thick ]  (112);
\path[->] (122) edge [thick ]  (222);

\path[->] (122) edge [thick ]  (112);
\path[->] (222) edge [thick ]  (212);

\path[->] (212) edge [thick ]  (211);
\path[->] (111) edge [thick ]  (112);
\path[->] (121) edge [thick ]  (122);
\path[->] (222) edge [thick ]  (221);

\end{scope}
\end{tikzpicture}
\caption{A 3-player 2-action game (left) and its corresponding best-response graph (right). The $\bullet$ vertices are sinks. The $\circ$ vertices form a cycle.}
\label{fig:intro_example}
\end{figure}

\begin{enumerate}[topsep=0ex,itemsep=0ex,leftmargin=0cm,label=(\roman*)]
\item We classify games according to the connectivity properties of their best-response graphs. Our interest in such a classification stems from the fact that the behaviours of many game dynamics are determined by such connectivity properties. A game's best-response graph is a directed graph whose vertex set is the set of pure action profiles and whose directed edges correspond to best-responses \citep*{young1993evolution}. An example is shown in Figure \ref{fig:intro_example}. A game's strict pure Nash equilibria correspond to the sinks of its best-response graph. Well-known classes of games categorized by the connectivity properties of their best-response graphs include \emph{weakly acyclic} games, i.e.\ those for which every vertex of the best-response graph can reach a sink along a directed best-response path, and \emph{acyclic} games, i.e.\ those whose best-response graphs contain no cycles. 
Acyclic games include the extensively studied class of potential games on which learning dynamics tend to perform well, while weak acyclicity has found wide interest because it is a necessary condition for the equilibrium convergence of certain (e.g.\ best-response) dynamics.
We introduce two new classes of games: \emph{connected} and \emph{super-connected} games. We say that a game is connected (super-connected) if its best-response graph has a sink and every non-sink can reach every sink (non-source) via best-response paths. The logical relationships between these game classes are shown below:
\begin{center}
    \begin{tikzpicture}
        \matrix (m) [matrix of math nodes,row sep=1em,column sep=1em,minimum width=0em] {
         \text{acyclic} & \text{weakly acyclic} \\
                        & \text{connected}      \\
                        & \text{super-connected} \\
        };
        \path[-stealth]
        (m-1-1) edge[double] (m-1-2)
        (m-3-2) edge[double] (m-2-2)
        (m-2-2) edge[double] (m-1-2);
    \end{tikzpicture}
\end{center}
The game in Figure \ref{fig:intro_example}, for example, is not acyclic but it is super-connected (and therefore connected and weakly acyclic). 

\item We build on recent results in probabilistic combinatorics to enumerate games and thereby quantify the relative sizes of the game classes shown above. Throughout, we take players' preferences to be encoded ordinally via a preference relation rather than cardinally via a utility function, thus ensuring that the space of (finite) games is countable (and finite for a fixed number of players and number of actions per player). We say that a game is `generic' if there are no indifferences in players' preferences. Our central result on game connectivity is this:
\begin{myquote}
    \emph{Among generic games that have a pure Nash equilibrium, all but a small fraction are connected. Moreover, the small fraction vanishes exponentially rapidly in the number of players}.
\end{myquote}
We prove the above result analytically for the case in which the number of players is much bigger than the maximum number of actions available to each player (Theorem~\ref{thm:main_game}). When the number of actions per player is fixed, our result implies that the fraction of generic games with a pure Nash equilibrium that are not connected vanishes exponentially as the number of players gets large.

Since connectedness is a ubiquitous property of generic games that have a pure Nash equilibrium, so is weak acyclicity. In contrast, acyclicity is very rare: we show that the fraction of generic games with a pure Nash equilibrium that are acyclic vanishes super-exponentially in the number of players (Proposition \ref{prop:acyclic_games}).
Finally, we show that super-connectedness exhibits a type of `phase transition' (Proposition \ref{prop:234}): the fraction of generic 2-action and 3-action games with a pure Nash equilibrium that are super-connected tends to 1 as the number of players gets large, but the fraction of generic $k$-action games that are super-connected tends to zero for any fixed $k \geq 4$.

\item Our characterization of the above game classes gives us insights into the `typical' structure of games, and this has important implications for adaptive dynamics. Indeed, we use our enumeration of connectivity properties in games to show (Theorem \ref{thm:conv}) that:
\begin{myquote}
 \emph{There is a simple adaptive dynamic that is guaranteed to lead to a pure Nash equilibrium in all but a small fraction of generic games that have one}.   
\end{myquote}
By `simple dynamic' we mean a dynamic that is uncoupled (i.e.\ a player's strategy depends only on the actions of other players and on their own preferences), stationary (i.e.\ time-independent), and 1-recall (i.e.\ no more than the last period's play is available to each player), and when we say that a dynamic is `guaranteed to lead to a pure Nash equilibrium' we mean that, starting at any action profile, the period-by-period play almost surely reaches a pure Nash equilibrium in finite time and, once there, never leaves it. \cite*{young2004strategic} shows that the best-response dynamic with inertia, which is a simple adaptive dynamic, is guaranteed to lead to a pure Nash equilibrium in every weakly acyclic game.\footnote{Under the best-response dynamic with inertia, in each period, each player $i$ independently best-responds to the current environment with probability $p_i \in (0,1)$ and does not update their action with probability $1-p_i$.} Our result follows from the observation that this is trivially also true in connected games. Our result contrasts with the aforementioned impossibility result of \cite*{hart2003uncoupled,hart2006stochastic} which states that there is no simple adaptive dynamic that is guaranteed to lead to a pure Nash equilibrium in every (generic) game that has one. Our result does not overturn this impossibility, but it limits its scope.

We also show that our result extends to some uncoupled dynamics that are not `simple', including regret-based and payoff-based dynamics. For example, we conclude that there is a regret-based dynamic that is guaranteed to lead to a pure Nash equilibrium in all but a small fraction of generic games that have one.

Our approach to the problem of finding adaptive dynamics that converge to pure Nash equilibria is to study the connectivity properties of games rather than to study the properties of the dynamics themselves. This approach can deliver novel conclusions, but it naturally also has limitations. An important caveat to our connectivity-based analysis, for example, is that we do not address convergence time to equilibrium since the latter depends on the specific dynamic used and on the game on which it is deployed. The final section of our paper offers a discussion of our approach, limitations, and directions for further work.

\paragraph*{Roadmap.} We define games and notions of game connectivity in Sections \ref{sec:games} and \ref{sec:connect}. Our main results on game connectivity are in Section \ref{sec:main_results}, and their implications for adaptive dynamics are discussed in Section \ref{sec:dynamics}. We conclude in Section \ref{sec:conc}. All proofs and further technical details are in Appendices \ref{sec:tech_conn}-\ref{sec:simulations}.

\end{enumerate}

\section{Games}\label{sec:games}
In this section we recall some standard definitions from the theory of games and introduce our notation.
For $n\in \N$, we use $[n]$ as shorthand for the set $\{1, \dots, n\}$. For each $a\in \N^n$ and $i\in [n]$, we write $a_{-i}$ for the element of $\N^{n-1}$ obtained by deleting the $i$th coordinate of $a$. In an abuse of notation, for $x\in \N$ and $a_{-i}\in \N^{n-1}$, we write $(x,a_{-i})$ for the element of $\N^n$ obtained by inserting $x$ into the $i$th coordinate of $a_{-i}$.

A \emph{game} is a tuple
\[
    \Big([n], \big([k_i]\big)_{i \in [n]} , (\succsim_i)_{i \in [n]}\Big),
\]
where $n\geq 2$ is an integer, $k_i\geq 2$ is an integer for each $i$, and for each $i$, $\succsim_i$ is a total preorder (i.e.\ a complete and transitive binary relation) on $A\coloneqq \prod_{i\in[n]}{[k_i]}$.
We say that $[n]$ is the \emph{player set} of the game and that $[k_i]$ is the \emph{action set} of player~$i$. Elements of $A$ are called \emph{action profiles}, and $\succsim_i$ is known as $i$'s \emph{preference relation}. For each $i$, let $\succ_i$ denote the asymmetric part of $\succsim_i$.

 An action $a_i$ of player $i$ is a \emph{best-response} to $a_{-i}$ if $(a_i, a_{-i}) \succsim_i (x, a_{-i})$ for every $x \in [k_i]$. An action profile $a \in A$ is a \emph{pure Nash equilibrium} if for each player $i \in [n]$, $a_i$ is a best-response to $a_{-i}$.

\section{Notions of game connectivity}\label{sec:connect}
The \emph{best-response graph} of a game is the directed graph $( A, \rightarrow )$ whose vertex set is the set of action profiles $A$ and whose directed edge set $\rightarrow$ is defined such that for $a,b\in A$,
\begin{myquote}
    $a \rightarrow b$ if and only if there exists $i\in[n]$ such that $a_{-i}=b_{-i}$, $b_i$ is a best-response to $a_{-i}$, and $b \succ_i a$.
\end{myquote}
In other words, there is a directed edge from $a$ to $b$ whenever $b_i$ is a strict best-response to $a_{-i}=b_{-i}$ for some player $i$.

We now define various classes of games in terms of the connectivity properties of best-response graphs.
As part of these definitions we will use standard terminology from the theory of directed graphs which we briefly recall here. Given a directed graph $(V,\rightarrow)$ with vertex set $V$ and edge set $\rightarrow$, a vertex $v \in V$ is a \emph{sink} if it has no outgoing edges (so corresponds to a strict pure-strategy Nash equilibrium), and a \emph{non-sink} otherwise. Similarly, a vertex $v \in V$ is a \emph{source} if it has no incoming edges, and a \emph{non-source} otherwise. For any pair of vertices $v,v' \in V$, we say that $v$ can \emph{reach} $v'$ if there is a sequence $(v^1,\dots,v^m)$ of vertices with $v^1 = v$ and $v^m = v'$ such that $v^i \rightarrow v^{i+1}$ for all $i\in[m-1]$; in this case we also say that the vertex~$v'$ can \emph{be reached from} $v$. Note that every vertex can reach and be reached from itself. A \emph{cycle} is a sequence $(v^1,\dots,v^m)$ of distinct vertices that has length $m$ at least $2$ and that satisfies $v^m\rightarrow v^1$ and $v^i \rightarrow v^{i+1}$ for all $i\in[m-1]$.

\begin{definition}
    A game is \emph{acyclic} if its best-response graph has no cycles.
\end{definition}

\begin{definition}
    A game is \emph{weakly acyclic} if its best-response graph has the property that every vertex can reach a sink.
\end{definition}
Observe that, by definition, a weakly acyclic game necessarily has at least one sink, and that acyclic games are weakly acyclic, but the converse need not hold.

Acyclicity and weak acyclicity are standard concepts (see e.g.\ \citealp*{fabrikant2013structure}) though they sometimes appear under different names in the literature.\footnote{For example, \cite*{takahashi2002pure} refer to weak acyclicity as quasi-acyclicity.} In our paper, the terms acyclicity and weak acyclicity follow the terminology of \citet*{young1993evolution}, who introduced the concept of weak acyclicity to the literature on dynamics in games.

Acyclic games are a superset of the very widely studied class of potential games \citep*{monderer1996potential}. Potential games have been the subject of intense research, particularly because many dynamics are guaranteed to converge to a pure Nash equilibrium in such games (e.g.\ \citealp*{hofbauer2002global,roughgarden2016twenty}). Weakly acyclic games are also very widely studied because weak acyclicity is a necessary condition for the guaranteed convergence of best-response dynamics to a pure Nash equilibrium from any starting vertex (e.g.\ see \citealp*{fabrikant2013structure,apt2015classification}). 

This paper introduces two further notions of connectivity.
\begin{definition}
    A game is \emph{connected} if its best-response graph has at least one sink and every non-sink can reach every sink.
\end{definition}

\begin{definition}
    A game is \emph{super-connected} if its best-response graph has at least one sink and every non-sink can reach every non-source.
\end{definition}
We note that super-connectedness implies connectedness, and connectedness implies weak acyclicity but, in each case, the converse need not hold. Moreover, as shown in Figure \ref{fig:a-b}, super-connectedness neither implies nor is implied by acyclicity.

\begin{figure}
    \centering
    \centering
    \begin{tabular}{ccc}
\begin{tikzpicture}[scale=3] 
\foreach \x in {1,2}
\foreach \y in {1,2}
\foreach \z in {1,2}
{\node[] (\z\x\y) at (\x,\y,\z) {$\circ$};} 
\node[] () at (1,2,1) {$\bullet$};
\node[] () at (2,1,2) {$\bullet$};
\node[] () at (1,2,2) {\color{magenta} $\bullet$};
\draw[thick, magenta] (1, 2, 2) circle (2pt);

\path[->] (211) edge [thick ]  (111);
\path[->] (121) edge [thick ]  (221);

\path[->] (111) edge [thick ]  (121);
\path[->] (211) edge [thick ]  (221);

\path[->] (212) edge [thick ]  (112);
\path[->] (122) edge [thick ]  (222);

\path[->] (122) edge [thick ]  (112);
\path[->] (222) edge [thick ]  (212);

\path[<-] (212) edge [thick ]  (211);
\path[->] (111) edge [thick ]  (112);
\path[->] (121) edge [thick ]  (122);
\path[->] (222) edge [thick ]  (221);

\end{tikzpicture}
&\vspace{0.1cm} &
\begin{tikzpicture}[scale=3] 
\foreach \x in {1,2}
\foreach \y in {1,2}
\foreach \z in {1,2}
{\node[] (\z\x\y) at (\x,\y,\z) {$\circ$};} 
\node[] () at (1,2,1) {$\bullet$};
\node[] () at (2,1,2) {$\bullet$};

\path[->] (211) edge [ultra thick, magenta ]  (111);
\path[->] (121) edge [thick ]  (221);

\path[->] (111) edge [ultra thick, magenta ]  (121);
\path[->] (211) edge [thick ]  (221);

\path[->] (212) edge [thick ]  (112);
\path[->] (122) edge [ultra thick, magenta ]  (222);

\path[->] (122) edge [thick ]  (112);

\path[->] (212) edge [ultra thick, magenta ]  (211);
\path[->] (111) edge [thick ]  (112);
\path[->] (121) edge [ultra thick, magenta ]  (122);
\path[->] (222) edge [thick ]  (221);

\path[->] (222) edge [ultra thick, magenta ]  (212);
\end{tikzpicture}
\\
(a) & & (b)
\end{tabular}
\caption{(a) Acyclic but not super-connected because some vertices, like the one shown in magenta, cannot reach every non-source. (b) Super-connected but not acyclic because every non-sink can reach every non-source but there is a cycle, shown in magenta.}
\label{fig:a-b}
\end{figure}

\section{Main results}\label{sec:main_results}

We quantify the relative sizes of the game classes defined in Section~\ref{sec:games} for generic games. We say that a game is \emph{generic} if for every $i$, and distinct action profiles $a$ and $a'$ that differ only in the $i$th index, either $a \succ_i a'$ or $a' \succ_i a$. To characterize the relative sizes of the game classes, we prove results regarding the prevalence of best-response graph properties among generic games and we derive their asymptotic implications as the number of players gets large.

Given an integer $n\geq 2$ and $\mbf{k}=(k_1,\dots,k_n)\in\{2,3,\dots\}^n$, we use $\cG(n,\mbf{k})$ to denote the set of all generic games with player set $[n]$ in which, for every $i \in [n]$, player $i$ has action set $[k_i]$. Since we are working with ordinal games, for a fixed $n$ and $\mbf{k}$, the set $\cG(n,\mbf{k})$ is finite.

The following is our main result on game connectivity.
\begin{theorem}\label{thm:main_game}
    There exist $c,\delta>0$ such that for all integers $n\geq 2$ and all $\mbf{k}\in\{2,3,\dots\}^n$, if $\max_i k_i \leq \delta \sqrt{n/ \log(n)}$ then
    \[
        \frac{| \{g\in\cG(n,\mbf{k})\colon g \text{\textnormal{ is connected}}\} |}{| \{ g\in\cG(n,\mbf{k})\colon g \text{\textnormal{ has a pure Nash equilibrium}}\} |} \geq 1 - e^{- cn}.
    \]
\end{theorem}

This result shows that, strikingly, connectedness is a ubiquitous property among generic games that have a pure Nash equilibrium and sufficiently many players relative to the maximal number of actions. Moreover, since connectedness implies weak acyclicity, the same is true of the latter property as well.

While the possible dependence of $k_i$ on $n$ is suppressed in our notation, Theorem~\ref{thm:main_game} allows for the number of actions per player to be growing with $n$ provided that the condition $\max_i k_i \leq \delta\sqrt{n/\log(n)}$ continues to be met. Of course, if the number of actions per player were fixed, our `sufficiently many players' condition would simplify to $n$ exceeding some constant. We investigate the regime in which the `sufficiently many players' condition fails in a recent companion paper \citep*{johnston2026gameconnectivityadaptivedynamics} and the results there are different: if the number of players is fixed but the number actions per player $k_i=k$ gets large, the limiting fraction of generic games with a pure Nash equilibrium that are connected is strictly less than one; see the companion paper for further details.

Theorem~\ref{thm:main_game} has important implications for adaptive dynamics in games, on which we elaborate in Section~\ref{sec:dynamics}.

Our next result is regarding acyclicity.
\begin{proposition}\label{prop:acyclic_games}
    There exists $c>0$ such that for all integers $n\geq 2$ and all $\mbf{k}\in\{2,3,\dots\}^n$, we have
    \[
        \frac{| \{g\in\cG(n,\mbf{k})\colon g \text{\emph{ is  acyclic}}\} |}{| \{g\in\cG(n,\mbf{k})\colon g \text{\emph{ has a pure Nash equilibrium}}\} |} \leq  e^{- cn2^n}.
    \]
\end{proposition}
Together, Theorem~\ref{thm:main_game} and Proposition~\ref{prop:acyclic_games} imply that there is a `split' in  game properties: among generic games that have a pure Nash equilibrium and sufficiently many players, acyclic games are very rare, while connected games and weakly acyclic games are very common. Note that since acyclic games are a superset of potential games, this also implies that potential games are very rare among generic games that have a pure Nash equilibrium and sufficiently many players.

Our final main result concerns super-connectedness. Let $\mbf{2}=(2,\dots,2)$, and similarly define $\mbf{3}$ and $\mbf{4}$.
\begin{proposition}\label{prop:234}
For $\mbf{k}=\mbf{2}$ or $\mbf{k}=\mbf{3}$ there exists $c>0$ such that for all integers $n\geq 2$,
    \[
        \frac{| \{g\in\cG(n,\mbf{k})\colon g \text{\textnormal{ is super-connected}}\} |}{| \{ g\in\cG(n,\mbf{k})\colon g \text{\textnormal{ has a pure Nash equilibrium}}\} |} \geq 1 - e^{- cn}.
    \]
    However, for each $\mbf{k} = (k,\dots,k) \geq \mbf{4}$, the fraction above tends to 0 as $n \to \infty$.
\end{proposition}
This shows that super-connectedness is a ubiquitous property of generic 2-action games and generic 3-action games that have a pure Nash equilibrium and sufficiently many players.\footnote{For this result we consider only games in which every player has the same number of actions.} However, this is not true of $k$-action games for $k\geq 4$. In fact, for $k \geq 4$, the fraction of generic $k$-action games with a pure Nash equilibrium that are super-connected vanishes as $n \to \infty$.

This shows two things. First, super-connectedness is too strong to be a typical property of generic games with many players. Unlike connectedness, it is not true that all but a small fraction of generic games that have a pure Nash equilibrium are super-connected. Second, connectivity properties that hold for small $\mbf{k}$ do not necessarily extend to large $\mbf{k}$. This is important because one cannot rely on results established for small numbers of actions as a guide for what to expect when the number of actions is large.

Our main results above are proved as corollaries of stronger results concerning the likelihood of analogous conditions holding in certain random directed graphs. The statements of these technical results on random graphs and the proofs themselves are in the appendix. The appendix also contains a section on the tightness of our main results. We provide a high-level discussion of our proof approach, as well as intuition for Theorem \ref{thm:main_game} and Propositions \ref{prop:acyclic_games} and \ref{prop:234}, in Section \ref{subsec:approach} below. 

In Appendix~\ref{sec:simulations}, we show via simulation that in practice all of the above results take hold even for generic games with relatively few players and few actions per player.

\subsection{Better-response graphs}
The results above concerned best-response graphs, but we can also draw similar conclusions for better-response graphs.

An action $a_i$ of player $i$ is a \emph{better-response} than $a_i'$ to $a_{-i}$ if $(a_i,a_{-i}) \succ_i (a_i',a_{-i})$. The \emph{better-response graph} of a game is the directed graph $( A, \rightarrow )$ whose vertex set is the set of action profiles $A$ and whose directed edge set $\rightarrow$ is defined such that for $a,b\in A$,
\begin{myquote}
    $a \rightarrow b$ if and only if there exists $i\in[n]$ such that $a_{-i}=b_{-i}$ and $b_i$ is a better-response to $a_{-i}$ than $a_i$.
\end{myquote}

For each connectivity property $P \in \{$acyclic, weakly acyclic, connected, super-connected$\}$, we say that a game is \emph{better-response $P$} if its better-response graph has that property. For example, a game is better-response connected if its better-response graph has a sink and the property that every non-sink can reach every sink. Observe that generalised ordinal potential games are precisely those that are better-response acyclic \citep*{monderer1996potential, fabrikant2013structure}.\footnote{A game $g=([n], ([k_i])_{i \in [n]} , (\succsim_i)_{i \in [n]})$ is a generalised ordinal potential game if there exists a function $\rho : A \rightarrow \mathbb{R}$ such that for each $i \in [n]$ and each pair of distinct action profiles $a$ and $a'$ that differ in only the $i$th index, $a \succ_i a'$ implies $\rho(a) > \rho(a')$. The game is an ordinal potential game or, simply, a potential game if, additionally, $\rho(a) > \rho(a')$ implies $a \succ_i a'$.}

Since a game's best-response graph is a subgraph of its better-response graph, we obtain the following logical relationships.

\noindent
\begin{tikzpicture}[every node/.style={font=\small}]
    \matrix (m) [matrix of math nodes,row sep=1em,column sep=1em,minimum width=0em] {
        \text{better-response acyclic} & \text{acyclic} & \text{weakly acyclic} & \text{better-response weakly acyclic}\\
                                &                & \text{connected}      & \text{better-response connected}\\
                                &                & \text{super-connected}& \text{better-response super-connected}\\
    };
    \path[-stealth]
    (m-1-1) edge[double] (m-1-2)
    (m-1-2) edge[double] (m-1-3)
    (m-1-3) edge[double] (m-1-4)
    (m-2-3) edge[double] (m-2-4)
    (m-2-3) edge[double] (m-1-3)
    (m-2-4) edge[double] (m-1-4)
    (m-3-3) edge[double] (m-2-3)
    (m-3-3) edge[double] (m-3-4)
    (m-3-4) edge[double] (m-2-4);
\end{tikzpicture}
The implications are now straightforward. Among generic games that have a pure Nash equilibrium, connected games, weakly acyclic games and their better-response counterparts, are very common, while acyclic games and better-response acyclic games are very rare. Analogous conclusions can similarly be drawn for super-connectedness.

\subsection{Classes of games with positive asymptotic density}
We now show that our results on game connectivity can be extended to any class of games $\mathcal{X}(n,\mbf{k}) \subseteq \{g\in\cG(n,\mbf{k})\colon g \text{ has a pure Nash equilibrium}\}$ that has positive asymptotic density, by which we mean that there is a $p \in (0,1]$ such that
    \[
     \frac{| \mathcal{X}(n,\mbf{k}) |}{| \{ g\in\cG(n,\mbf{k})\colon g \text{ has a pure Nash equilibrium}\} |} \geq p
    \]
    for all sufficiently large $n$.

The following is a corollary of Theorem \ref{thm:main_game}.

\begin{corollary}\label{cor:nzm}
If $\mathcal{X}(n,\mbf{k}) \subseteq \{g\in\cG(n,\mbf{k})\colon g \textnormal{ has a pure Nash equilibrium}\}$ has positive asymptotic density then the fraction of games in $\mathcal{X}(n,\mbf{k})$ that are connected gets close to 1 for sufficiently large $n$.
\end{corollary}

Here is an example. \cite{rinott2000number} show that for any integer $z \geq 0$,
    \[
        \frac{| \{g\in\cG(n,\mbf{k})\colon g \text{ has exactly $z$ pure Nash equilibria}\} |}{| \cG(n,\mbf{k}) |} \to \frac{e^{-1}}{z!}
    \]
as $n \to \infty$ or as $k_i \to \infty$ for at least two players $i$.
From this we infer that for any integer $z \geq 1$, the set 
\[
\{g\in\cG(n,\mbf{k})\colon g \text{ has at least $z$ pure Nash equilibria}\}
\]
has positive asymptotic density, from which we can conclude that for any integer $z\geq 1$, all but a vanishing fraction of generic games that have exactly $z$ pure Nash equilibria are connected.

\subsection{Relationship to \texorpdfstring{\cite*{amiet2021pure}}{Amiet et al. (2021)}}\label{subsec:amiet}
Our paper is related to \cite*{amiet2021pure}. The focus of that paper is not entirely on game connectivity, but it contains a result that is related to our work and we discuss this relationship here.

Consider a vertex $v$ in the best-response graph of a game in $\cG(n,\mbf{k})$. We say that the game is $v$-\emph{connected} if its best-response graph has at least one sink and the property that if $v$ is a non-sink, then it can reach every sink. Similarly, we say that the game is $v$-\emph{super-connected} if its best-response graph has at least one sink and the property that if $v$ is a non-sink, then it can reach every non-source.

Expressed in the language of our paper, the arguments of \citeauthor*{amiet2021pure} imply that there exists $c>0$ such that for all integers $n\geq 2$ and any vertex $v$,
\[
    \frac{| \{g\in\cG(n,\mbf{2})\colon g \text{\textnormal{ is $v$-super-connected}}\} |}{| \{ g\in\cG(n,\mbf{2})\colon g \text{\textnormal{ has a pure Nash equilibrium}}\} |} \geq 1 - e^{- cn}.
\]
This differs from our results in two main ways. First, a game is {(super-)connected} if it is $v$-(super-)connected for \emph{every} vertex $v$, so (super-)connectedness is very much stronger than $v$-(super-)connectedness. Indeed, there are typically almost $2^n$ vertices that are non-sinks and connectedness requires $v$-connectedness to hold \emph{simultaneously} for all of them. Second, the result of \citeauthor*{amiet2021pure} applies only to two-action games, while our results apply much more generally. As we have seen, connectivity properties that hold for games with few actions per player may not extend to games with many actions per player, and these games are fundamentally different.

To be clear, while our results on game connectivity generalise one result found in \citeauthor*{amiet2021pure}, our paper is not a generalisation of that paper. Importantly, \citeauthor*{amiet2021pure} obtain results for non-generic games, which we do not have.

\subsection{Proof approach and intuition for the main results}\label{subsec:approach}
Suppose that a game $G$ is drawn uniformly at random from $\cG(n,\mbf{k})$. Then
\[
\mathbb{P}(\text{$G$ has property $P$} \mid \text{$G$ has a pure Nash equilibrium}) 
\]
is equal to the fraction
\[
\frac{| \{g\in\cG(n,\mbf{k})\colon g \text{\textnormal{ has property $P$}}\} |}{| \{ g\in\cG(n,\mbf{2})\colon g \text{\textnormal{ has a pure Nash equilibrium}}\} |} .
\]
Rather than directly enumerating games in $\cG(n,\mbf{k})$ that have certain properties, we rely on the simple insight above to instead draw games uniformly at random and work out the probability that such randomly drawn games have certain properties.\footnote{Our paper thus contributes methodologically to the literature on `random games'. The distribution of pure Nash equilibria was studied in \cite*{goldberg1968probability}, \cite*{dresher1970probability}, \cite*{powers1990limiting}, and \cite*{stanford1995note}. Further results relating to the number of Nash equilibria also appear in \cite*{mclennan1997maximal,mclennan2005expected}, \cite*{barany2007nash}, \cite*{daskalakis2011connectivity}, and \cite*{pei2023nash}. See also \cite*{alon2021dominance} for dominance-solvable games, and \cite*{mimun2024best,collevecchio2024basins} and \cite*{ashkenazigolan2025simultaneous} for best-response dynamics.} The latter is simpler and yields our desired quantities of interest.

We now give a high-level explanation of some elements of the proofs of our main results. All the proofs for the results in Section \ref{sec:main_results} (and of even stronger results) are in the appendix.

\paragraph{Intuition for Theorem \ref{thm:main_game}} 
 First, draw a game $G$ uniformly at random from $\cG(n,\mbf{k})$, and observe that its best-response graph has a natural product structure: each vertex (action profile) lies on exactly $n$ lines, one per coordinate direction, and for each line, all the edges in the line point toward the unique `winning' action profile; namely, the action profile that corresponds to a best-response for player $i$ given the other players' actions.

At the heart of our proof is an argument showing the likely existence of a large strongly connected component in the best-response graph of $G$; in other words, we find a large set of vertices which can all reach each other along directed paths. Indeed, we define a vertex in the best-response graph of $G$ to be \emph{good} if it wins in roughly the expected number of coordinates, so that it has both a substantial number of incoming and a substantial number of outgoing edges, and we show via a concentration argument that the vast majority (all but an exponentially small fraction) of vertices are good when $n$ is large. Moreover, good vertices form a giant strongly connected component: having many outgoing edges means each good vertex can reach many other vertices, having many incoming edges means each can be reached from many vertices, and the expansion properties of the underlying grid structure of the graph ensure that best-response paths between these large sets of vertices exist with high probability. This part of our proof is based on work of \cite*{mcdiarmid2021component}, who study the component structure of random subgraphs of the undirected hypercube graph.

Once we know that with high probability all good vertices can reach each other, we are then left with `plugging in' the remaining, unusual, vertices. We plug in these vertices by building on arguments in \cite*{bollobas1993connectivity}. 
Broadly, we show that those with at least one outgoing edge can reach many vertices, among which good vertices appear with high probability. Those which can be reached from a moderate number of vertices (and, in particular, sinks) are conversely reachable from many vertices, among which good vertices appear with high probability.

Theorem~\ref{thm:main_game} follows from an even stronger result, Theorem~\ref{thm:main_grids}, which is stated and discussed in Appendix \ref{sec:tech_conn}, and is proved in Appendices \ref{sec:proof_intro}-\ref{sec:main_proof}. We note that our proof employs different methods from the ones used in \cite*{amiet2021pure}. For the result that we discussed in Section \ref{subsec:amiet}, \citeauthor*{amiet2021pure} also draw games uniformly at random, in their case from $\cG(n,\mbf{2})$. But, because there are only two actions per player, all edges in the resulting best-response graphs are independent of each other, and their proofs rely heavily on this feature. Once there are more than two actions per player, the edges are no longer independent, and this requires an alternative approach. Our approach here is also different from the approach in our recent companion paper \citep*{johnston2026gameconnectivityadaptivedynamics} where we consider the setting where the number of players is fixed (or growing slowly) and the number of actions gets large.\footnote{In the `many players' regime considered here, each vertex of the best-response graph is incident to at least $n$ edges, so $n$ being large is likely to contribute to greater connectivity. If instead it is the number of actions that is large, then most vertices remain incident to a fixed number of edges, and fully characterising the `many actions' regime requires different arguments from those given here.}

\paragraph{Intuition for Proposition \ref{prop:acyclic_games}} Cycles are very likely to be present in the best-response graph of a randomly drawn game $G$. Consider a \emph{plane} of action profiles that results from fixing the actions of all but two players. We show that a cycle exists in such a plane with probability at least $1/8$ (a constant independent of the number of actions or players). The fact that there are many planes when $n$ is large implies that it is extremely unlikely that none of them contains a cycle. Our proof is in Appendix \ref{sec:tech_acyclic}.

\paragraph{Intuition for Proposition \ref{prop:234}}
As explained above, if a vertex can be reached from a moderate number of vertices, it can be reached from every non-sink, so for super-connectedness to fail, there must be a vertex that can only be reached by a small number of vertices.
Let us call a vertex a \emph{near-source} if it wins exactly one of its lines, which implies that it has only $k-1$ incoming edges (one per other vertex on the won line). If, other than the near-source, all the vertices on the near-source's won line are themselves sources, we say the near-source is \emph{isolated}. An isolated near-source cannot be reached by any non-sink that does not lie on its won line, implying that super-connectedness fails. 

Conditional on a vertex $v$ being a near-source, the probability that each of the other $k-1$ vertices on its won line are sources is $((k-1)/k)^{(n-1)(k-1)}$. The probability that $v$ is a near-source is given by $(n/k) ((k-1)/k)^{(n-1)}$ since it must lose all but exactly one of its lines, and there are $n$ lines going through $v$. There are $k^n$ vertices in the best-response graph, so the expected number of isolated near-source vertices is:
\[
k^n \cdot \frac{n}{k} \left(\frac{k-1}{k}\right)^{(n-1)} \cdot \left(\frac{k-1}{k}\right)^{(n-1)(k-1)} = n \cdot \left(\frac{(k-1)^k}{k^{k-1}}\right)^{n-1} .
\]
The limit in $n$ is governed entirely by the ratio $(k-1)^k / k^{k-1}$. The ratio is below 1 for $k \in \{2,3\}$ and exceeds 1 for $k \geq 4$, so isolated near-sources are vanishingly rare when $k<4$ but their number grows exponentially when $k\geq4$. The final step to showing that super-connectedness fails for $k \geq 4$, is to compute the second moment of the number of isolated near-sources to show the number is concentrated around its mean.

Of course, showing that there are no isolated near-sources with high probability is not enough to show that super-connectedness occurs with high probability. However, letting $K\coloneq \max_i(k_i)$, we show in Theorem~\ref{thm:main_grids} that every vertex that can be reached from $(1+\eps)K \log(K)$ vertices, can be reached from every non-sink, and for $k = 2$, winning a single line is enough to reach this threshold. For $k = 3$, the threshold is slightly too large for winning a single line to reach it, but a slight modification of the proof allows us to reduce the threshold, at the expense of allowing the failure probability to depend on $K$. This modification leads to Theorem~\ref{thm:constant_K}, which is stated in the appendix, and is enough to prove the $k = 3$ case.

\section{Implications for adaptive dynamics in games}\label{sec:dynamics}

We now consider the implications of our results regarding game connectivity for games played over time according to \emph{adaptive dynamics}. We begin by recalling some standard notions.

First, a player $i$'s \emph{observation set} at time $t$, denoted $o_i^t$, consists of all the information that~$i$ can observe at time $t$. Precisely what objects enter into this set varies depending on the regime under consideration, and below it will be made clear which regimes we are considering. For each integer $k\geq 2$, let $O_k$ denote the set of all possible observation sets (under the given regime) for a player with action set $[k]$. A \emph{strategy} for a player with action set $[k]$ is a function $f\colon O_k \rightarrow \Delta([k])$, where $\Delta([k])$ is the probability simplex over $[k]$. Let $n\geq 2$ and $k_1,\dots,k_n\geq 2$ be integers, and write $\mbf{k}=(k_1,\dots,k_n)$. A \emph{dynamic} on $\cG(n,\mbf{k})$ consists of a specification of what information enters into each player's observation set at each time, and a strategy $f_i$ with action set $[k_i]$ for each player $i$.

The play of a game $g\in \cG(n,\mbf{k})$ under a given dynamic begins at time $t=0$ at an initial action profile $a^0$ chosen arbitrarily. This informs each player's observation set $o_i^1$ according to the dynamic. At time $t=1$, each player updates their action (randomly) according to $f_i(o_i^1)$, and we denote the new (random) action profile by $a^1$. The play continues in this manner, with each player
updating their action at $t=2$ according to $f_i(o_i^2)$ to produce an action profile $a^2$, and so on.

Employing standard terminology, we say that a dynamic is \emph{uncoupled} if at each time $t$, each player $i$'s observation set contains (at most) their own preference relation $\succsim_i$ and the ordered history of play $a^0,\dots,a^{t-1}$. For an integer $m\geq 1$, an uncoupled dynamic is $m$\emph{-recall} if at each time~$t$, each player $i$'s observation set contains (at most) the current time $t$, their own preference relation $\succsim_i$, and the ordered history of play $a^{t-m},\dots,a^{t-1}$ for the past $m$ steps, or the full history of play if $t<m$. An uncoupled and $m$-recall dynamic is \emph{stationary} if, at each time $t$, each player $i$'s observation set consists of their own preference relation $\succsim_i$, and the ordered history of play $a^{t-m},\dots,a^{t-1}$ for the past $m$ steps, or the full history of play if $t<m$. Crucially, for $t\geq m$, the only information about the current time $t$ available to the players is that $t\geq m$, so their strategies become time-independent after this point.

\begin{definition}
    A dynamic is \emph{simple} if it is uncoupled, 1-recall, and stationary. That is, a dynamic is simple if at each time $t$, player $i$'s observation set contains at most their own preference relation $\succsim_i$ and last period's action profile $a^{t-1}$.
\end{definition}

We consider the following strong notion of convergence to a pure Nash equilibrium.
\begin{definition}\label{def:convergence}
    A dynamic on $\cG(n,\mbf{k})$ \emph{converges almost surely to a pure Nash equilibrium} of a game $g\in \cG(n,\mbf{k})$ if when $g$ is played according to the dynamic from any initial action profile, almost surely there exists $T < \infty$ and a pure Nash equilibrium $a^*$ of $g$ such that $a^t=a^*$ for all $t \geq T$.
\end{definition}

\subsection{The possibility of convergence to a pure Nash equilibrium}
As mentioned in the introduction, the following impossibility result is well-known.

\begin{theorem*}[\citealp*{hart2006stochastic,jaggard2014self}]
For all $n\geq 3$ and $\mbf{k}\in \N^n$ with $k_i\geq 2$ for all $i$ (or $k_i\geq 3$ for all $i$ if $n=3$), there is no simple dynamic on $\cG(n,\mbf{k})$ for which play converges almost surely to a pure Nash equilibrium in every game in $\cG(n,\mbf{k})$ that has one. 
\end{theorem*}
\noindent Observe that the impossibility applies to games of almost any size, including games with many players. It is instructive to revisit a proof of this result. Figure \ref{fig:counter} shows the best-response graph of a game considered in \cite{hart2006stochastic}. The graph's key feature is that each vertex in magenta has exactly one out-going edge. Any simple dynamic on this game initiated at one of the magenta vertices cannot get to any vertex other than the magenta vertices.\footnote{The reason is this. Take a vertex $v$ in the magenta cycle whose only outgoing edge is in direction $j$. By uncoupledness, changing player $j$'s preferences should not affect what $i \neq j$ does at $v$. So change $j$'s preferences so that $v$ becomes the unique Nash equilibrium. If the dynamic must converge to a pure Nash equilibrium, $i \neq j$ must not move at $v$ in this modified game. But this implies that $i\neq j$ must not move at $v$ in the original game either.} In particular, any simple dynamic initiated at one of the magenta vertices cannot reach the pure Nash equilibrium. A game with a feature like this one can be embedded into a larger game, whether in terms of the number of players or the number of actions.\footnote{Best-response graphs with features like the one shown in Figure \ref{fig:counter} are rare. In a typical game, a vertex winning all but one of its lines is very rare when the number of lines is large, and it is exponentially unlikely that there two such vertices connected by an edge.}

\begin{figure}
\centering
    \begin{tikzpicture}[scale=2.2] 
\foreach \x in {1,2,3}
\foreach \y in {1,2,3}
\foreach \z in {1,2,3}
{\node[] (\z\x\y) at (\x,\y,\z) {$\circ$};} 
\node[] () at (3,1,3) {$\bullet$};

\node[] () at (2,3,1) {\color{magenta}$\bullet$};
\node[] () at (1,3,2) {\color{magenta}$\bullet$};
\node[] () at (2,3,2) {\color{magenta}$\bullet$};

\node[] () at (1,2,1) {\color{magenta}$\bullet$};
\node[] () at (2,2,1) {\color{magenta}$\bullet$};
\node[] () at (1,2,2) {\color{magenta}$\bullet$};

\path[->] (113) edge [thick]  (213);
\path[->] (313) edge [thick]  (213);
\path[->] (123) edge [thick]  (223);
\path[->] (123) edge [ultra thick, magenta]  (223);
\path[->] (323) edge [thick]  (223);
\path[->] (133) edge [thick]  (233);
\path[->] (333) edge [thick]  (233);

\path[->] (312) edge [thick]  (112);
\path[->] (212) edge [ultra thick, magenta]  (112);
\path[->] (322) edge [thick]  (122);
\path[->] (332) edge [thick]  (132);

\path[->] (311) edge [thick]  (211);
\path[->] (111) edge [thick]  (211);
\path[->] (321) edge [thick]  (121);
\path[<-] (331) edge [thick]  (131);

\path[->] (113) edge [thick]  (123);
\path[->] (133) edge [thick]  (123);
\path[<-] (213) edge [thick]  (233);
\path[<-] (213) edge [ultra thick, magenta]  (223);
\path[->] (313) edge [thick]  (323);
\path[->] (333) edge [thick]  (323);

\path[->] (112) edge [thick]  (122);
\path[->] (112) edge [ultra thick, magenta]  (122);
\path[->] (132) edge [thick]  (122);
\path[<-] (212) edge [thick]  (232);
\path[<-] (312) edge [thick]  (332);

\path[->] (111) edge [thick]  (121);
\path[->] (131) edge [thick]  (121);
\path[<-] (211) edge [thick]  (231);
\path[->] (311) edge [thick]  (331);

\path[->] (111) edge [thick]  (112);
\path[->] (113) edge [thick]  (112);
\path[->] (121) edge [thick]  (123);
\path[->] (122) edge [ultra thick, magenta]  (123);
\path[->] (131) edge [thick]  (132);
\path[->] (133) edge [thick]  (132);

\path[->] (211) edge [thick]  (212);
\path[->] (213) edge [thick]  (212);
\path[->] (213) edge [ultra thick, magenta]  (212);
\path[->] (221) edge [thick]  (223);
\path[->] (231) edge [thick]  (233);

\path[->] (311) edge [thick]  (312);
\path[->] (313) edge [thick]  (312);
\path[->] (321) edge [thick]  (323);
\path[<-] (331) edge [thick]  (333);

\end{tikzpicture}

\emph{Note}. Some edges of the graph are omitted to keep the illustration uncluttered.
\caption{Best-response graph of a generic game with a pure Nash equilibrium. Any simple dynamic initiated at one of the magenta vertices cycles through those vertices forever.}
\label{fig:counter}
\end{figure}

In contrast with the above impossibility result, we show below that if $n$ grows (much more quickly than $\max_i k_i$), there is a simple dynamic on $\cG(n,\mbf{k})$ that converges almost surely to a pure Nash equilibrium on all but a vanishingly small fraction of generic games in the class that have one.

\begin{theorem}\label{thm:conv}
For $n$ sufficiently large relative to $\max_i k_i$, there is a simple dynamic on $\cG(n,\mbf{k})$ for which play converges almost surely to a pure Nash equilibrium in all but an exponentially small fraction of games in $\cG(n,\mbf{k})$ that have one. 
\end{theorem}
\pagebreak
\noindent The proof is straightforward. The \emph{best-response dynamic with inertia} is defined as follows:\footnote{This dynamic is well-known and versions of it appear in, for example, \cite{young2009learning} and \citet*{swenson2018distributed}. The manner in which ties might be broken among multiple best-responses in non-generic games is immaterial for our purposes.}  at each step $t$, independently of the other players, each player $i$ sets $a_i^t$ to be a best-response to $a_{-i}^{t-1}$ with some fixed probability $p_i \in (0,1)$ and sets $a_i^t= a_i^{t-1}$ with complementary probability $1-p_i$. \cite{young2004strategic} showed that, for any choice of parameters $p_i\in(0,1)$, this dynamic converges almost surely to a pure Nash equilibrium in every weakly acyclic game (Observe that the dynamic allows for simultaneous moves, but all that matters for the convergence result is that paths in which players move sequentially receive positive probability.) Since every connected game is weakly acyclic, from Theorem~\ref{thm:main_game} we obtain that there exists $c>0$ such that for integers $n\geq 2$ and $\mbf{k}\in\{2,3,\dots\}^n$, if $n$ is sufficiently large relative to $\max_i k_i$, then the proportion of games in
    \[
        \{g\in \cG(n,\mbf{k})\colon g \textnormal{ has a pure Nash equilibrium}\}
    \]
for which the best-response dynamic with inertia converges almost surely to a pure Nash equilibrium is at least $1-e^{-cn}$.\footnote{This result also holds for variants of the best-response dynamic with inertia. For example, it straightforwardly holds for the better-response dynamic with inertia. It also holds for a one-at-a-time version of the best-response dynamic in which, at each step $t$, exactly one player $i$ is selected at random from among all players to update their action, and this player plays a best-response to $a_{-i}^{t-1}$. Convergence properties of this one-at-a-time version were investigated by \cite*{heinrich2023best} via simulation.} Noting that the best-response dynamic with inertia is simple completes the proof. 

We now briefly comment on how Theorem \ref{thm:conv} relates to the broader literature on possibility and impossibility results for dynamics in games. Establishing impossibility for a class of dynamics typically consists of finding collections of games such that no dynamic in the class is guaranteed to lead to equilibrium in all of them. Naturally, this hinges on the parameters of the problem, namely, (i) the information that is allowed to determine players' decisions in the dynamic, (ii) the notion of convergence that is required, (iii) the type of equilibrium to which the dynamic converges, and (iv) the class of games to which the dynamic is applied. The impossibility  of \cite*{hart2006stochastic} and \cite*{jaggard2014self} that we stated above concerns (i) simple dynamics with (ii) period-by-period play converging almost surely to a (iii) pure Nash equilibrium in (iv) generic games that have one. But this impossibility is by no means the only one.\footnote{Recent examples: \cite*{milionis2023impossibility} and \cite{schipper2022strategic}.} Much of the literature has focused on relaxing one or several of the requirements (i)-(iv), thereby giving rise to possibility results, or has strengthened some of the requirements while weakening others, thereby characterizing where the `frontier' lies between the possible and the impossible for dynamics in games.\footnote{There are several possibility results for dynamics that lead to mixed Nash or correlated equilibria. For example, there are uncoupled and completely uncoupled dynamics for which the empirical distribution of play converges almost surely to the set of correlated (or coarse correlated) equilibria in all games \citep{foster1997calibrated,fudenberg1999conditional,hart2000simple,hart2001reinforcement}. There are also uncoupled and completely uncoupled dynamics for which the behaviour probabilities converge almost surely to a mixed Nash equilibrium in all generic games \citep{foster2006regret,germano2007global}. \cite*{vlatakis2020no}, however, show that a commonly studied regret-based dynamic does not lead to mixed Nash equilibria. There are stationary, 2-recall, uncoupled dynamics for which the period-by-period play converges almost surely to a pure Nash equilibrium in all games that have one \citep{hart2006stochastic,cesa_bianchi2006prediction}. \citet*{jaggard2014self} identify other uncoupled dynamics with this convergence property in a bounded-recall synchronous setting. \citet{babichenko2012completely} shows that there is no completely uncoupled dynamic for which the period-by-period play converges almost surely to a pure Nash equilibrium in every generic game that has one. On the other hand, there is a completely uncoupled dynamic for which the period-by-period play is at a pure Nash equilibrium `most of the time' in all generic games that have one \citep{young2009learning,pradelski2012learning}.} In contrast, we maintain all the requirements of \cite*{hart2006stochastic} and \cite*{jaggard2014self}, but instead of focusing on the impossible/possible dichotomy, Theorem \ref{thm:conv} \emph{quantifies} the proportion of games for which the impossibility holds; namely, it is exponentially small in the number of players.

\subsection{The scope of existing results on adaptive dynamics}\label{sec:scope}

Our results on game connectivity allow us to quantify the scope of existing results regarding the convergence properties of adaptive dynamics. 

Consider the following (far from exhaustive) list of results on adaptive dynamics:\footnote{E.g.\ see \cite{newton2018evolutionary} for further examples.}
\begin{itemize}
    \item \cite{young1993evolution} shows that `adaptive play' (a class of dynamics involving inertia and finite memory) converges almost surely to a pure Nash equilibrium in all weakly acyclic games.
    \item \cite{friedman2001learning} shows that the `better-reply dynamic with sampling' converges almost surely to a pure Nash equilibrium in all better-response weakly acyclic games.
    \item \cite*{marden2007regret} shows that a regret-based dynamic converges almost surely to a pure Nash equilibrium in all (generic) better-response weakly acyclic games.
    \item \cite*{marden2009payoff} show that a purely payoff-based dynamic leads to play that is at a pure Nash equilibrium in every better-response weakly acyclic game `most of the time'.
\end{itemize}
Since connected games are (better-response) weakly acyclic, it follows from Theorem~\ref{thm:main_game} that all the above results in fact hold in all but a small fraction of generic games that have a pure Nash equilibrium.

Observe that the above examples cover a wide range of dynamics, including ones that are regret-based and payoff-based (and defined in a cardinal setting rather than an ordinal one like ours) and that are not necessarily `simple' according to our definition (since regret-based dynamics are ordinarily not 1-recall, for example). Moreover, the notions of convergence employed in the aforementioned papers are sometimes different from almost sure convergence of period-by-period play.\footnote{Almost-sure convergence is a strong notion of convergence since it requires period-by-period play to eventually settle on, and never leave, a pure Nash equilibrium. We have focused on it, in part, because this is the notion that the impossibility results of \cite{hart2006stochastic} and \cite*{jaggard2014self} pertain to. Weaker notions of convergence (for examples, see \citealp{young2004strategic}) allow for possibility results that are different from ours. For example, \cite{young2009learning} shows that so-called `trial-and-error learning' is an uncoupled (in fact, completely uncoupled) dynamic that, for any $\eps >0$, is at a pure Nash equilibrium for a $1-\eps$ proportion of time steps in any generic game that has one. This is a powerful result because it applies to \emph{every} generic game but the notion of convergence there is weaker than almost-sure convergence of period-by-period play. With the latter, once a pure Nash equilibrium is reached, it is never left, whereas trial-and-error learning requires constant experimentation so there is always a positive probability of leaving a pure Nash equilibrium and wandering before settling on one again.} For example, the notion of convergence in \cite*{marden2009payoff} requires the dynamic to be at a pure Nash equilibrium `most of the time', which is different from almost sure convergence.

\subsection{Dynamics in connected vs weakly acyclic games}

Observe that Theorem~\ref{thm:main_game} is stronger than what we needed for Theorem~\ref{thm:conv} or for the implications for adaptive dynamics that we described in Section \ref{sec:scope}. In all cases, it would have been sufficient to have established the ubiquity of weak acyclicity rather than of connectedness. This raises the question of whether there are situations in which connectedness has implications for dynamics that do not follow from weak acyclicity alone.

There are at least two situations in which connectedness plays an essential role: (i) stochastic stability under evolutionary dynamics, and (ii) design. 

\begin{enumerate}[topsep=0ex,itemsep=0ex,leftmargin=0cm,label=(\roman*)]
\item When there are multiple equilibria, it is natural to ask which of these equilibria will be played. An approach taken in evolutionary game theory is to determine at which of the equilibria perturbed dynamics will spend most of their time; these are adaptive dynamics in which players' choices are subject to random errors parametrised by some $\eps >0$.\footnote{A perturbed version of the best-response dynamic with inertia might specify that, at each time, any updating player plays a best-response with probability $1-\eps$ and, with complementary probability $\eps > 0$, selects an action uniformly at random.} The `stochastically stable' states of such a dynamic---the states at which the dynamic spends most of its time---are the action profiles that are assigned positive probability as $\eps \to 0$ in the invariant distribution of the Markov process induced by the dynamic.

Several methods have been proposed to work out the stochastically stables states, such as the minimum-cost tree technique or the radius-coradius technique \citep*{kandori1993learning,kandori1995evolution,young1993evolution,freidlin2012random,ellison2000basins}. In all cases, one assigns a `cost' of moving from one action profile to another and calculates the total cost of transitions from one Nash profile to another. The cost of any best-response move is always set to zero in such calculations, and is positive otherwise. The Nash profiles that are relatively the most costly to escape and cheapest to return to are the stochastically stable states. In connected games, since there is a best-response path from any non-Nash profile to every Nash profile, the cost of moving from a non-Nash profile to any Nash profile is \emph{zero} (observe that this is not generally true in weakly acyclic games). Therefore, in a connected game, the total cost of transitioning from Nash equilibrium $a$ to Nash equilibrium $a'$ is just the cost of the initial deviation that kicks the system out of $a$. This local property, which \cite{newton2024conventions} refer to as the `one-shot' property, is very easy to calculate (and therefore side-steps having to do a full tree calculation à la \citealt*{freidlin2012random}, for example).

\cite{newton2024conventions} recently leveraged the one-shot property of stochastically stable states in connected games to show that, in typical games with a large number of players, different evolutionary dynamics select different types of Nash equilibria according to the welfare properties of the equilibria.

\item In a connected game, a planner who can determine the order with which players get a chance to update their actions can always steer a best-response dynamic from any non-Nash profile to any Nash equilibrium, and therefore to their most preferred one. Such steering to \emph{any} Nash profile is not generally possible in weakly acyclic games. This is a somewhat abstract point, but it serves to illustrate a substantive wedge between connectedness and weak acyclicity.\footnote{The idea of steering dynamics to desirable equilibria (though not by determining order of play) is discussed, for example, in \cite*{balcan2011leading}.}
\end{enumerate}

\section{Discussion}\label{sec:conc}

Let us return to our analogy of `beyond the worst-case' analysis of algorithms. As we have argued, simple adaptive dynamics perform poorly (i.e.\ are not guaranteed to lead to a pure Nash equilibrium) on worst-case inputs of the type given in the impossibility result of \citeauthor*{hart2003uncoupled}, but they perform very well on `typical' inputs because all but a small fraction of generic games with a pure Nash equilibrium have a property (namely, connectedness) which is conducive to convergent dynamics. While this provides a hopeful message, there are important caveats.

First, there is our notion of what a `typical' input is. Each possible configuration of players' preferences is given equal weight in our enumeration of generic games.
One consequence of this is that some well-known classes of games are `rare' according to our enumeration. For example, since potential games have acyclic best-response graphs \citep*{monderer1996potential}, our result on the prevalence of acyclic games implies that potential games are very rare. However, this is still an important class of games: they are an appropriate model for certain types of strategic interaction such as congestion \citep*{rosenthal1973class}. Our uniform enumeration reflects the fact that it is not a priori obvious which game structures ought to be given more or less weight when considering which of them are `typical' in the context of adaptive dynamics.\footnote{A player's outcomes can depend in complex ways on the actions of other players, particularly when there are many players. In such complex situations, players may not attempt to take account of all aspects of the strategic interaction before making their choices as they might in, say, an experimental setting, but may instead simply resort to `local' improving moves; in other words, in complex environments, one can reasonably expect players to resort to simple decision heuristics (e.g.\ better- and best-response dynamics) like the ones considered in this paper \citep{sandholm2010population}.} If one were nevertheless interested in questions of connectivity in restricted classes of games, we believe that a probabilistic combinatorial approach like ours may be usefully applied but would take us beyond the scope of our paper. It may be fruitful to investigate the connectivity properties of, for example, anonymous games \citep{gradwohl2021large},\footnote{A game is anonymous if, at any action profile, a player's ranking of their actions depends only on the numbers of other players selecting particular actions but not on the identities of those other players.} or games featuring local interactions, such as graphical games \citep*{kearns2007graphical,daskalakis2011connectivity,kearns2013graphical} or action-graph games \citep*{jiang2011action}.

Second, as mentioned in the introduction, since our focus is not on any specific dynamic, we have not addressed the question of the speed of convergence to equilibrium \citep{arieli2016stochastic}, and this is important when it comes to finding equilibria in games. In fact, it is known that adaptive dynamics can take a long time to converge \citep{hart2010long}.\footnote{We expect convergence to take exponential time for the best-response dynamic with inertia.} That said, greater knowledge of connectivity properties may help to establish general results on the speed of convergence in some classes of games.\footnote{Convergence can be fast in potential games (e.g.\ see \citealp*{awerbuch2008fast}) and in anonymous games (e.g.\ see \citealp{babichenko2013best}).}

Many other open questions remain. For example, further analysis of game connectivity may shed light on the relationship between connectivity and \emph{completely} uncoupled dynamics \citep{babichenko2012completely},\footnote{Completely uncoupled dynamics have even lower informational requirements than uncoupled ones: a dynamic is completely uncoupled if, at each time $t$, each player's observation set contains only their own realised utility payoffs and their own past actions.}. Also, our focus was on generic games but there are, to our knowledge, few results on the connectivity properties of non-generic games.\footnote{Every game $([n], ([k_i])_{i \in [n]} , (\succsim_i)_{i \in [n]})$ has a utility-based representation $([n], ([k_i])_{i \in [n]} , (u_i)_{i \in [n]})$ with, for each player $i$, a utility function $u_i :A \rightarrow \mathbb{R}$ representing their preference relation $\succsim_i$. Our notion of genericity is equivalent to the condition that such a utility-based game has no payoff ties, i.e.\ that for each $i$ and any distinct profiles $a$ and $a'$ that differ only in the $i$th index, $u_i(a) \neq u_i(a')$. Furthermore, any utility-based game for which the utility numbers are perturbed by small random shocks independently drawn from an atomless distribution is almost surely generic in our sense of the term. Generic games are therefore common. That said, non-genericity arises in certain natural situations, such as in normal-form representations of sequential games. While \cite*{amiet2021pure} and \cite*{collevecchio2024finding} derive some connectivity properties of non-generic games with two actions per player there are, to our knowledge, no results for non-generic games with more than two actions per player.}

\clearpage

\appendix

\section{Connectivity of directed grids}\label{sec:technical}\label{sec:tech_conn}

Theorem~\ref{thm:main_grids}, the main result of this section and the central technical contribution of our paper, is about the connectivity properties of random subgraphs of directed Hamming graphs. We first introduce our notation and then state the theorem, before explaining how Theorem~\ref{thm:main_game} follows.

For $n\in \N$ and $\mbf{k}=(k_1,\dots,k_n)\in \{2, 3, \dots\}^n$, the \emph{Hamming graph} $H(n,\mbf{k})$ is the graph with vertex set $V(n,\mbf{k})\coloneqq \prod_{i=1}^n[k_i]$ and edges between $n$-tuples precisely when they differ in exactly one coordinate. For $i\in[n]$, a \emph{line of $V(n,\mbf{k})$ in coordinate $i$} is a subset of $V(n,\mbf{k})$ of size $k_i$ whose elements pairwise differ in exactly the $i$th coordinate. A \emph{line} of $V(n,\mbf{k})$ is a subset which is a line of $V(n,\mbf{k})$ in coordinate $i$ for some $i$. Note that a line induces a complete subgraph of $H(n,\mbf{k})$. The \emph{directed Hamming graph} $\vv{H}(n,\mbf{k})$ is the simple directed graph formed by replacing each edge $uv$ of $H(n,\mbf{k})$ with directed edges $u\rightarrow v$ and $v\rightarrow u$.

Let $\Lnk$ be the random subgraph of the directed Hamming graph defined by independently and uniformly at random choosing a \emph{winner} among the vertices of each line of $\vv{H}(n,\mbf{k})$, and within that line keeping only those edges $u \rightarrow v$ whose endpoint,~$v$, is the winner. Observe that in this random subgraph, each line induces a directed star in which all edges are oriented towards the winner. Our interest in $\Lnk$ stems from the following.

\begin{remark}
    The graph $\Lnk$ has the same distribution as the best-response graph of a game drawn uniformly at random from amongst all games in $\cG(n,\mbf{k})$.
\end{remark}

As in Theorem~\ref{thm:main_game}, we will study these objects when $n$ is large relative to $\max_i(k_i)$; our proof breaks down when this is not the case. We now state our main theorem.

\begin{restatable}{theorem}{main}\label{thm:main_grids}
    For all $\eps>0$ there exist $c,\delta>0$ such that for all integers $n\geq 2$ and all $\mbf{k}\in\{2,3,\dots\}^n$, if $K\coloneq \max_i(k_i)$ satisfies $K \leq \delta \sqrt{n/ \log(n)}$, then with failure probability at most $\prod_{i=1}^{n}{k_i}^{-c}$, every vertex of $\Lnk$ can either be reached from at most $N \coloneqq(1+\eps)K\log(K)$ vertices, or from every non-sink.
\end{restatable}

Theorem~\ref{thm:main_grids} shows that there is a sharp dichotomy in the number of vertices which can reach each vertex in the random graph $\vec{L}(n,\mathbf{k})$: with high probability, every vertex $v$ is one of exactly two types. Either $v$ can only be reached by a small number of vertices (explicitly, $N = (1+\varepsilon)K\log(K)$), or it can be reached from every non-sink. There is no intermediate possibility. The threshold $N$ is small relative to the total number of vertices (which grows exponentially in $n$), so the ``small'' category is genuinely small even when $K$ is increasing with $n$.

The full proof of Theorem~\ref{thm:main_grids} is postponed to Appendices~\ref{sec:proof_intro},~\ref{sec:foothold},~\ref{sec:strong_cmpnt}, and~\ref{sec:main_proof}.  In Appendix~\ref{sec:tech_tight} we examine the tightness (or lack thereof) of various aspects of Theorem~\ref{thm:main_grids}. In particular, we show that neither the failure probability nor the value of $N$ can be significantly improved in general.

We use the remainder of this subsection to explain how Theorem~\ref{thm:main_game} follows from Theorem~\ref{thm:main_grids}. The key observation is that sinks fall into the category of vertices that have maximal in-neighbourhood size: a sink wins every coordinate line so it accumulates incoming edges from many vertices (at least $n+1$) and, in particular, more than the threshold $N$ (the latter is upped-bounded by $n$ given our assumption on $K$). Hence every sink can be reached from every non-sink, which is precisely connectedness. To derive this implication more formally, we highlight the following corollary of Theorem~\ref{thm:main_grids}, in which we denote by $R_{n,\mbf{k}}$ the event that every non-sink in $\Lnk$ can reach every sink, and by $S_{n,\mbf{k}}$ the event that $\Lnk$ has at least one sink.

\begin{corollary}\label{cor:conn}
    There exist $c_0,c_1>0$ and $\delta \in (0,1]$ such that for all integers $n\geq 2$ and all $\mbf{k}\in\{2,3,\dots\}^n$, if $K\coloneq \max_i(k_i)$ is such that $K \leq \delta \sqrt{n/ \log(n)}$, then
    \begin{enumerate}[label=\textnormal{(\alph*)}]
        \item $\Prb(R_{n,\mbf{k}}) \geq 1 - e^{-c_0 n}$,
        \item $\Prb(R_{n,\mbf{k}}\, |\, S_{n,\mbf{k}}) \geq 1 - e^{-c_1 n}$.
    \end{enumerate}
\end{corollary}
\begin{proof}
    Let $c$ and $\delta'$ be as given by Theorem~\ref{thm:main_grids} in the case $\eps =1$, and let $\delta$ be the minimum of~$1$ and $\delta'$.
    Then for $n$, $\mbf{k}$, and $K$ as in the statement of the corollary, we have that with failure probability at most $\prod_{i=1}^{n}{k_i^{-c}}$ every vertex of $\Lnk$ can either be reached from at most $2K\log(K)$ vertices or from every non-sink. However, we have $2K\log(K)\leq K^2\log(K)\leq n$, so if this event holds then every non-sink can reach every sink, because all sinks can be reached from at least $n+1$ vertices. Finally, note that $\prod_{i=1}^{n}{k_i^{-c}}\leq 2^{-cn}\leq e^{-c_0 n}$ for some $c_0>0$, which proves part (a).

    Next,
    \[
        \Prb( R_{n,\mbf{k}} \, | \, S_{n,\mbf{k}} )
        = \frac{\Prb(R_{n,\mbf{k}} \cap S_{n,\mbf{k}})}{\Prb(S_{n,\mbf{k}})}
        \geq \frac{ \Prb(R_{n,\mbf{k}} ) - (1-\Prb( S_{n,\mbf{k}} )) }{ \Prb(S_{n,\mbf{k}})}
        \geq 1 - \frac{e^{-c_0 n}}{\Prb(S_{n,\mbf{k}})},
    \]
    where we used part (a) in the final step. It follows from work in \cite{rinott2000number} that there exists a positive universal constant which lower bounds $\Prb(S_{n,\mbf{k}})$ for all $n\geq 2$ and all $\mbf{k}\in \{2,3,\dots\}^{n}$, completing the proof of part (b).
\end{proof}

As remarked above, $\Lnk$ has the same distribution as the best-response graph of a game drawn uniformly at random from among all games in $\cG(n,\mbf{k})$. Because our draws are uniform, we have that
\[
\Prb( R_{n,\mbf{k}} \, | \, S_{n,\mbf{k}} ) =
\frac{| \{g\in\cG(n,\mbf{k})\colon g \text{ is connected}\} |}{| \{ g\in \cG(n,\mbf{k}) \colon g \text{ has a pure Nash equilibrium}\} |},
\]
so Theorem~\ref{thm:main_game} follows immediately from part (b) of Corollary~\ref{cor:conn}.

\begin{proof}[Proof of Corollary \ref{cor:nzm}]
Replace $S_{n,\mbf{k}}$ in the proof of Corollary~\ref{cor:conn} part (b) by the event that $\Lnk$ is in $\mathcal{X}(n,\mbf{k})$. 
\end{proof}

\subsection{Outline of the proof of Theorem~\ref{thm:main_grids}}\label{sec:outline}

We detail the proof of Theorem~\ref{thm:main_grids} in Appendices~\ref{sec:proof_intro},~\ref{sec:foothold},~\ref{sec:strong_cmpnt}, and~\ref{sec:main_proof}. We first start by providing a high-level overview of our arguments that is a bit more technical than the one we provided for Theorem~\ref{thm:main_game} in the main text.

At the heart of our proof is an argument showing the likely existence of a certain large strongly connected component in $\Lnk$; in other words, we find a large set of vertices which can all reach each other along directed paths. Indeed, we define a vertex of $\Lnk$ to be \emph{good} if the number of lines that it wins is close to the expected number, and then show that with high probability every good vertex can reach all other good vertices along directed paths. Since most vertices win close to the expected number of lines, this connects up a good proportion of the graph and it remains to `plug in' the remaining, unusual, vertices. This part of our proof is based on work of \cite*{mcdiarmid2021component}, who study the component structure of random subgraphs of the undirected hypercube graph.

Once we know that all good vertices are in the same strongly connected component, it is sufficient to show that every non-sink $x$ can reach some good vertex $u$, and every vertex $y$ which can be reached from more than $N=(1+\eps)K\log(K)$ vertices can be reached from some good vertex $v$. Indeed, this yields a directed path from $x$ to $y$ via $u$ and $v$, where we utilise the strongly connected component to get from $u$ to $v$. The first step towards this is to `establish a foothold' by showing that with high probability (a) every non-sink can in fact reach at least $n/2$ vertices, and (b) every vertex that can be reached from more than $N$ vertices can in fact be reached from at least $n/2$ vertices. The inspiration for considering such an event comes from the work of~\cite*{bollobas1993connectivity}. The final step in the proof is then to show that if a vertex can reach or be reached from at least $n/2$ vertices, then it is very unlikely that none of these vertices is good.

\section{Proof of Theorem~\ref{thm:main_grids}: preliminaries}\label{sec:proof_intro}

Throughout the proof of Theorem~\ref{thm:main_grids}, i.e.\ throughout Appendices~\ref{sec:proof_intro},~\ref{sec:foothold},~\ref{sec:strong_cmpnt}, and~\ref{sec:main_proof}, we will take $n\geq 2$ to be an integer, we will take $\mbf{k}\in\{2,3,\dots\}^n$, and we will let $K\coloneqq \max_i(k_i)$.

\begin{definition}
We will describe a probability $p_{\eps}(n,\mbf{k})$ with parameters $\eps>0$, $n$, and $\mbf{k}$ as being \emph{very small} if for all $\eps$ there exist $c_{\eps},\delta_{\eps}>0$ depending only on $\eps$ such that $p_{\eps}(n,\mbf{k})\leq \prod_{i=1}^{n}{k_i^{-c_{\eps}}}$ for all $K\leq \delta_{\eps} \sqrt{n /\log (n)}$. For a probability $p(n,\mbf{k})$ with no dependence on $\eps$, the constants $c_{\eps}$ and $\delta_{\eps}$ should be replaced by universal constants.
\end{definition}

\begin{definition}
We will say that $p(n,\mbf{k})$ is \emph{extremely small} if there exist $c,\delta_0>0$ such that for all $\delta \in (0,\delta_0)$, if $K\leq \delta \sqrt{n/\log(n)}$, then $p(n,\mbf{k})\leq e^{-cn\log(K)/\delta}$. 
\end{definition}
Observe that every extremely small probability is also very small.

\begin{definition}
If $p_\eps(n,\mbf{k})$ or $p(n,\mbf{k})$ is very or extremely small, we will say that the complementary probability is very or extremely high, respectively.   
\end{definition}

\begin{definition}
We say that an event $F_\eps(n,\mbf{k})$ or $F(n,\mbf{k})$ occurs \emph{with very high probability} (wvhp) or \emph{with extremely high probability} (wehp) if the probability that it occurs is very or extremely high respectively.   
\end{definition}

Given this terminology, Theorem~\ref{thm:main_grids} is equivalent to the statement that wvhp, every vertex of $\Lnk$ can either be reached from at most $(1+\eps)K\log(K)$ vertices or from every non-sink. This motivates our definition of very small probabilities. Our definition of extremely small probabilities is motivated by the following lemma, which demonstrates that such probabilities are amenable to union bounds over $V(n,\mbf{k})$.

\begin{lemma}\label{lem:union_bound}
    If $p(n,\mbf{k})$ is an extremely small probability, then for all fixed $a>0$ the probability $K^{an}\cdot p(n,\mbf{k})$ is also extremely small.
\end{lemma}
\begin{proof}
    Let $c$ and $\delta_0$ witness the fact that $p(n,\mbf{k})$ is extremely small. Then for all $\delta\in(0,\delta_0)$, if $K \leq \delta \sqrt{n/\log(n)}$, then
    \[
        K^{an}\cdot p(n,\mbf{k})
        \leq e^{an\log(K)-cn\log(K)/\delta}
        = e^{n\log(K)(a-c/\delta)}.
    \]
    Let $\delta'_0\in (0,\delta_0)$ be small enough that $c/(2\delta'_0)>a$, then for all $\delta\in(0,\delta'_0)$ we have $a-c/\delta < -c/(2\delta)$, so letting $c'=c/2$ we see that $c',\delta'_0$ witness the fact that $K^{an}\cdot p(n,\mbf{k})$ is extremely small.
\end{proof}

We will also make frequent use of the following simple result which follows from elementary analyses of the various cases.

\begin{lemma}\label{lem:add_probs}
    The sum of two very small probabilities is very small and the sum of two extremely small probabilities is extremely small.
\end{lemma}

Before starting the proof Theorem~\ref{thm:main_grids} in earnest, we record the following two standard results which will be useful at various points. For a discussion of these results (and much more), we refer the reader to~\cite{frieze2015introduction}.

\begin{lemma}[Chernoff bound]
    \label{lem:chernoff}
    Let $X_1,\dots,X_n$ be independent Bernoulli random variables, let $X=\sum_{i=1}^n X_i$, and let $\mu=\mathbb{E}[X]$. Then for all $\eps\geq 0$ we have \[\Prb(X\leq(1-\eps)\mu)\leq e^{-\eps^2\mu/2}.\]
\end{lemma}

\begin{lemma}[Application of Markov's inequality]
    \label{lem:markov}
    Let $Z_1, \dots, Z_m$ be non-negative integer valued random variables, and suppose that $\sum_{i=1}^m\mathbb{E}[Z_{i}]\leq p$ for some $p\in [0,1]$. Then the probability that $Z_1 = \dotsb = Z_m  = 0$ is at least $1-p$.
\end{lemma}
\begin{proof}
    Let $Z = \sum_{i=1}^m Z_i$ and note that $\mathbb{E}[Z] = \sum_{i=1}^m\mathbb{E}[Z_{i}]\leq p$. By Markov's inequality, $\Prb(Z\geq 1)\leq \mathbb{E}[Z]\leq p$, and the complement of the event $\{Z \geq 1\}$ is $\{Z_1 = \dotsb = Z_m  = 0\}$.
\end{proof}

\section{Proof of Theorem~\ref{thm:main_grids}: establishing a foothold}
\label{sec:foothold}

Throughout this section we let $\eps>0$ and set $N=(1+\eps)K\log(K)$, as in the statement of Theorem~\ref{thm:main_grids}. We will also assume (without loss of generality) that $k_1\leq k_2\leq \dots\leq k_n$. Recall from Section~\ref{sec:outline} that one step in our proof of Theorem~\ref{thm:main_grids} will be to show that, wvhp, in $\Lnk$ all non-sinks can reach more than $n/2$ vertices, and all vertices which can be reached from more than $N$ vertices can be reached from more than $n/2$ vertices. We write $A_\eps$ for the event that this condition holds in $\Lnk$.

\begin{definition}[Event $A_\eps$]
    Let $A_\eps$ be the event that the following two conditions are satisfied:
    \begin{itemize}
        \item all non-sinks can reach more than $n/2$ vertices; and
        \item every vertex that can be reached from more than $N$ vertices can be reached from more than $n/2$ vertices.
    \end{itemize}
\end{definition}

This section is devoted to establishing the following lemma.

\begin{lemma}\label{lem:foothold}
    The event $A_\eps$ occurs with very high probability.
\end{lemma}
Our proof of Lemma \ref{lem:foothold} follows the approach used by \cite*{bollobas1993connectivity} to study a random subgraph of $\vv{H}(n,\mbf{2})$ with a similar distribution to that of ${\vv{L}(n, \mbf{2})}$. Explicitly, \citeauthor*{bollobas1993connectivity} consider the subgraph $\vv{Q}^n_p$ of $\vv{H}(n,\mbf{2})$ in which each directed edge is kept independently with probability $p$ (where  $p$ is close to $1/2$). In their model, two vertices differing in a single coordinate may be joined by no directed edges, one directed edge or even both directed edges. This is in contrast to ${\vv{L}(n, \mbf{2})}$ where such vertices are always joined by exactly one directed edge.
We remark that, although the models are quite different, the main result of \cite*{bollobas1993connectivity} is similar to Corollary~\ref{cor:constant_K} (see Section \ref{sec:tech_tight}) in the case that $K = 2$. Indeed, they show that with high probability as $n \to \infty$, every non-sink can reach every non-source in $\vv{Q}^n_{1/2}$, which is the same up to the model.

Imitating \citeauthor*{bollobas1993connectivity}, for each $1\leq m \leq \prod_{i=1}^n k_i$, define the random variable $X_m$ to be the number of vertices of $\Lnk$ which can reach exactly $m$ vertices (recall that every vertex can reach and be reached from itself). Analogously, let $Y_m$ be the number of vertices which can be reached from exactly $m$ vertices. Observe that $A_{\eps}$ can equivalently be defined as the event that $X_m=0$ for all $2\leq m\leq n/2$ and $Y_m=0$ for all $N< m\leq n/2$.

Given a set $S\subseteq V(n,\mbf{k})$, we say that a line of $V(n,\mbf{k})$ in coordinate $i$ is an \emph{incomplete line of $S$} if its intersection with $S$ has size other than $0$ or $k_i$.  If $v$ can reach exactly $m$ vertices, then running a depth first search from $v$ gives a tree $T$ with $m$ vertices, in which all edges are oriented away from $v$ and the winner of every incomplete line of~$T$ is in $T$. Similarly, if $v$ can be reached from exactly $m$ vertices, then we may build a tree~$T$ with $m$ vertices where all the edges are oriented towards $v$ and the winner of every incomplete line of $T$ is outside of $T$. It follows that $X_m$ and $Y_m$ are bounded above by the number of pairs $(v, T)$, where $T$ is an appropriate tree with $m$ vertices rooted at~$v$.

We will use the following folklore result to upper bound the numbers of such trees in $H(n,\mbf{k})$. A short combinatorial proof is given in~\cite*{mcdiarmid2021component}.

\begin{lemma}
    \label{lem:trees}
    If $G$ is a graph with maximum degree $\Delta$, then for each $m\in \N$ there are at most $(e\Delta)^{m-1}$ trees of order $m$ in $G$ that contain a given vertex.
\end{lemma}

When applied to $H(n,\mbf{k})$, Lemma~\ref{lem:trees} gives that there are at most $(enK)^{m-1}$ trees of order~$m$ in $H(n,\mbf{k})$ that contain a given vertex. Lemma~\ref{lem:foothold} follows from the next two lemmas, which handle the $X_m$ and $Y_m$ parts of the statement respectively.

\begin{lemma}\label{lem:X_small}
    With very high probability, $X_m=0$ for all $2\leq m\leq n/2$.
\end{lemma}
\begin{proof}
    We need to show that there exist universal $c,\delta>0$ such that if $K\leq \delta \sqrt{n/\log(n)}$, then $X_m=0$ for all $2\leq m\leq n/2$ with probability at least $1-\prod_{i=1}^{n}{k_i^{-c}}$.
    Thus, let $\delta>0$ be small and assume that $K\leq \delta \sqrt{n\log(n)}$.

    Fix $2\leq m\leq n/2$ and let $T$ be a tree of order $m$ in $H(n,\mbf{k})$. Given the discussion preceding Lemma~\ref{lem:trees}, we wish to upper bound the probability that the winner of every incomplete line of $T$ is in $T$. To this end, it will be helpful to lower bound the number of lines containing exactly one vertex of $T$. Each vertex of $T$ is in $n$ lines, so there are~$mn$ pairs $(u,l)$ consisting of a vertex $u$ in $T$ and a line $l$ containing it. For each pair of distinct vertices $u$ and $v$ in $T$, if $u$ and $v$ are contained in some common line $l$, then delete the pairs $(u,l)$ and $(v,l)$ from this set. Since any pair of vertices have at most one common line, this process removes at most $2\binom{m}{2}$ pairs from the set, and we deduce that there are at least $mn-m^2$ lines of $V(n,\mbf{k})$ which contain exactly one vertex of $T$.

    The winner of each of these lines is in $T$ independently. Since we want to upper bound the probability that the winner of all of these lines is in $T$, we may assume that they are all in as low a coordinate direction as possible (recall that $k_1\leq\dots\leq k_n$ by assumption). At most $m$ incomplete lines are in any given coordinate direction, so the probability that the winner of every incomplete line of $T$ is in $T$ is at most $\prod_{i=1}^{n-m}{k_i^{-m}}$.

    By Lemma~\ref{lem:trees}, the number of pairs $(v,T)$ where $v\in V(n,\mbf{k})$ and $T$ is a tree of order~$m$ in $H(n,\mbf{k})$ containing $v$ is at most $(enK)^{m-1} \cdot \prod_{i=1}^n{k_i}$, so by the discussion before that lemma we have
    \begin{equation*}
            \mathbb{E}[X_m]
            \leq \frac{(enK)^{m-1} \cdot \prod_{i=1}^n{k_i}}{\prod_{i=1}^{n-m}{k_i^m}}
            \leq \frac{K^m(enK)^{m-1} \cdot \prod_{i=1}^{n-m}{k_i}}{\prod_{i=1}^{n-m}{k_i^m}}
            \leq \frac{(enK^2)^m}{\prod_{i=1}^{n-m}{k_i^{m-1}}}.
    \end{equation*}

    Applying the fact that $m-1\geq m/2$ (since $m\geq 2$), we obtain
    \begin{equation*}
        \begin{split}
            \mathbb{E}[X_m]
            \leq \frac{(enK^2)^m}{\prod_{i=1}^{n-m}{k_i^{m/2}}}
            \leq \left(\frac{enK^2}{\prod_{i=1}^{n-m}{k_i^{1/2}}}\right)^m
            \leq \prod_{i=1}^{n-m}{k_i^{-m/3}}
        \end{split}
    \end{equation*}
    where the final inequality follows by taking $\delta$ small enough that $enK^2\leq 2^{n/12}$, which is at most $\prod_{i=1}^{n-m}{k_i^{1/6}}$ since $m\leq n/2$.

    \begin{claim}\label{claim:X}
        If $\delta$ is small enough, then $\prod_{i=1}^{n-m}{k_i^{-m/3}}\leq \prod_{i=1}^{n}{k_i^{-1/2}}$ for all $2\leq m\leq n/2$.
    \end{claim}
    \begin{proof}
        After rearranging, we need to show that $\prod_{i=n-m+1}^{n}{k_i^{1/2}}\leq \prod_{i=1}^{n-m}{k_i^{m/3-1/2}}$ for all $2\leq m\leq n/2$. The left-hand side of this inequality is at most $K^{m/2}$ and the right-hand side is at least $2^{(n-m)(m/3-1/2)}$. Raising both sides to the power of $2/m$, it is sufficient that $K\leq 2^{(n-m)(2/3-1/m)}$. The right-hand side of this inequality is at least $2^{n/12}$, and we can take $\delta$ small enough that $K\leq 2^{n/12}$, so the claim is proved.
    \end{proof}

    Applying the claim, we have
    \[
        \sum_{m=2}^{n/2}\mathbb{E}[X_m]\leq \frac{n}{2}\cdot\prod_{i=1}^{n}{k_i^{-1/2}}.    \]
    By taking $\delta$ to be sufficiently small we can ensure that this is at most $\prod_{i=1}^{n}{k_i^{-c}}$ for some $c>0$. Lemma~\ref{lem:markov} now yields that $X_m=0$ for all $2\leq m \leq n/2$ with failure probability at most~$\prod_{i=1}^{n}{k_i^{-c}}$, as required.
\end{proof}

The next lemma deals with the $Y_m$ part of Lemma~\ref{lem:foothold}. Note that Lemma~\ref{lem:foothold} follows immediately from Lemma~\ref{lem:add_probs}, Lemma~\ref{lem:X_small}, and Lemma~\ref{lem:Y_small}.

\begin{lemma}\label{lem:Y_small}
    With very high probability, $Y_m=0$ for all $N< m\leq n/2$.
\end{lemma}
\begin{proof}
    We need to show that there exist $c_\eps,\delta_\eps>0$ depending only on $\eps$ such that if $K\leq \delta_\eps \sqrt{n/\log(n)}$, then $Y_m=0$ for all $N< m\leq n/2$ with failure probability at most $\prod_{i=1}^{n}{k_i^{-c_\eps}}$. In fact, we will show a stronger failure probability of at most $e^{-c_\eps n\log(K)}$. Thus, let $\delta_\eps>0$ be small and assume that $K\leq \delta_\eps \sqrt{n/\log(n)}$.

    We will employ a similar strategy to that used to prove Lemma~\ref{lem:X_small}. Fix $N< m\leq n/2$ and let $T$ be a tree of order $m$ in $H(n,\mbf{k})$. We will upper bound the probability that the winner of every incomplete line of $T$ is not in $T$ using the lower bound of $mn - m^2$ on the number of incomplete lines of $T$ (from the proof of Lemma~\ref{lem:X_small}).
    The winner of each of these lines is in $T$ independently, so the probability that all the winners are outside~$T$ is at most $(1-1/K)^{m(n-m)}$.

    Hence, by Lemma~\ref{lem:trees} and the discussion preceding it, we have
    \begin{equation*}
        \begin{split}
            \mathbb{E}[Y_m]
            & \leq K^n \cdot (enK)^{m-1} \cdot \left(1-\frac{1}{K}\right)^{m(n-m)} \\
            & \leq \left[enK^2 \cdot \left(K^{1/m}\left(1-\frac{1}{K}\right)\right)^{n-m}\right]^m.
        \end{split}
    \end{equation*}
    Using that $m> N=(1+\eps)K\log(K)$ and $1+x\leq e^x$ for all $x$ we have
    \begin{equation}\label{eq:K_constant_different}
        K^{1/m}\left(1-\frac{1}{K}\right)
        \leq K^{1/(1+\eps)K\log(K)}e^{-1/K}
        = \exp\left(\frac{-\eps}{(1+\eps)K}\right).
    \end{equation}
    Assuming that $\delta_\eps \leq 1/2$, we have that $K^2\leq n$. Applying this and $m\leq n/2$ yields
    \[
        \mathbb{E}[Y_m]
        \leq \left[en^2\cdot \exp\left(\frac{-\eps(n-m)}{(1+\eps)K}\right)\right]^m
        \leq \left[en^2\cdot \exp\left(\frac{-\eps n}{2(1+\eps)K}\right)\right]^m.
    \]
    By making $\delta_\eps$ small enough that $\eps n/(4(1+\eps)K)\geq \log(en^2)$ and using the fact that $m\geq K\log(K)$, we have
    \[
        \mathbb{E}[Y_m]
        \leq \exp\left(\frac{-\eps mn}{4(1+\eps)K}\right)
        \leq \exp\left(\frac{-\eps n\log(K)}{4(1+\eps)}\right).
    \]
    Thus,
    \[
        \sum_{N<m\leq n/2}\mathbb{E}[Y_m]
        \leq \frac{n}{2}\cdot \exp\left(\frac{-\eps n\log(K)}{4(1+\eps)} \right) \leq \exp\left( - \frac{\eps}{8(1+\eps)} n \log(K)\right)
    \]
    for sufficiently large (depending only on $\eps$) $n$.
    By making $\delta_\eps$ sufficiently small relative to~$\eps$, we can ensure that we only need to consider values for $n$ which are sufficiently large, and we find that $\sum_{N<m\leq n/2}\mathbb{E}[Y_m]$ is at most $e^{-c_{\eps} n\log(K)}$ for some $c_\eps>0$ depending only on $\eps$.
    Lemma~\ref{lem:markov} now yields that $Y_m=0$ for all $N< m \leq n/2$ with failure probability at most~$e^{-c_\eps n\log(K)}$, as required.
\end{proof}

\section{Proof of Theorem~\ref{thm:main_grids}: a strongly connected component}
\label{sec:strong_cmpnt}

The next step in the proof of Theorem~\ref{thm:main_grids} will be to find a large strongly connected component in $\Lnk$.
To this end, we make the following two definitions.
\begin{definition}[Good vertex]
    A vertex of $\Lnk$ is \emph{good} if it wins at least $n/3K$ but at most $3n/4$ of its lines.
\end{definition}
\begin{definition}[Event $B$]
    Let $B$ be the event that every good vertex of $\Lnk$ can reach all other good vertices along directed paths, that is, all good vertices are in the same strongly connected component of $\Lnk$. 
\end{definition}
The main result of this section is that $B$ occurs wehp.

\begin{lemma}
    \label{lem:good_cmpnt}
    The event $B$ occurs with extremely high probability.
\end{lemma}

Our proof of Lemma~\ref{lem:good_cmpnt} is based on an approach taken by \cite*{mcdiarmid2021component}, and will proceed via the following two auxiliary lemmas.

\begin{lemma}
    \label{lem:good_nbr}
    With extremely high probability, every vertex in $V(n,\mbf{k})$ has an $H(n,\mbf{k})$-neighbour which is good in $\Lnk$.
\end{lemma}

\begin{lemma}
    \label{lem:good_3path}
    With extremely high probability, every good vertex in $\Lnk$ can reach all good vertices within $H(n,\mbf{k})$-distance 3 of it.
\end{lemma}

Given Lemmas~\ref{lem:good_nbr} and~\ref{lem:good_3path}, Lemma~\ref{lem:good_cmpnt} follows easily. This proof and the proof of Lemma~\ref{lem:good_nbr} are straightforward generalisations of work in~\cite*{mcdiarmid2021component}, but we include them for completeness.

\begin{proof} [Proof of Lemma~\ref{lem:good_cmpnt}]
    Let $\vv{F}$ be any realisation of $\Lnk$ for which the conclusions of Lemma~\ref{lem:good_nbr} and Lemma~\ref{lem:good_3path} both hold, and note that this happens wehp by those two lemmas and Lemma~\ref{lem:add_probs}. That is, let $\vv{F}$ be any outcome of $\Lnk$ in which every good vertex can reach all other good vertices at distance at most~3 from it in $H(n,\mbf{k})$, and in which every vertex has a neighbour in $H(n,\mbf{k})$ which is good. To prove Lemma~\ref{lem:good_cmpnt}, it suffices to show that for all good vertices $x$ and $y$, there exists a directed path in $\vv{F}$ from $x$ to $y$.

    Let $x$ and $y$ be good vertices of $\vv{F}$ and choose a path $P=p_0p_1\dots p_t$ in $H(n,\mbf{k})$ where $p_0=x$ and $p_t=y$. If $t\leq 3$, then there exists a directed path from $x$ to $y$ in $\vv{F}$. Otherwise $p_2,p_3,\dots,p_{t-2}$ have good $H(n,\mbf{k})$-neighbours $q_2,\dots,q_{t-2}$ respectively. Then $x$ and $q_2$ are at distance at most 3, $q_i$ and $q_{i+1}$ are at distance at most 3 for all $2\leq i\leq t-3$, and~$q_{t-2}$ and $y$ are at distance at most 3, so $\vv{F}$ contains a path from $x$ to $y$ via $q_2,\dots,q_{t-2}$.
\end{proof}

\begin{proof}[Proof of Lemma~\ref{lem:good_nbr}]
    For a vertex $v\in V(n,\mbf{k})$, let $X$ be the number of lines that $v$ wins in $\Lnk$ and note that $v$ is good if and only if $n/3K\leq X \leq 3n/4$. Since $X$ is a sum of $n$ independent Bernoulli random variables and has mean $\mu=\sum_{i=1}^n{1/k_i}\geq n/K$, by Lemma~\ref{lem:chernoff} we have
    \[
        \Prb\Big(X< \frac{n}{3K}\Big)
        \leq \Prb\Big(X\leq \frac{\mu}{3}\Big)
        \leq e^{-2\mu/9}
        \leq e^{-2n/(9K)}.
    \]
    Similarly, $n-X$ is a sum of $n$ independent Bernoulli random variables and has mean $\mu'\geq n/2$, so by Lemma~\ref{lem:chernoff} again we have
    \[
        \Prb\Big(X> 3n/4\Big)
        = \Prb\Big(n-X< n/4\Big)
        \leq \Prb\Big(n-X\leq \frac{\mu'}{2}\Big)
        \leq e^{-\mu'/8}
        \leq e^{-n/16}.
    \]
    It follows (using $K\geq 2$) that $v$ is good with failure probability at most $e^{-n/(10K)}$.

    Now fix $u\in V(n,\mbf{k})$ and pick one vertex other than $u$ from each of the $n$ lines containing it, say $v_1,\dots,v_n$. The $v_i$ are distinct and no two of them share a line, so they are good independently of one another. Hence, the probability that $u$ has no good $H(n,\mbf{k})$-neighbour is at most $e^{-n^2/(10K)}$, so by a union bound over $u$, the probability that there exists a vertex with no good $H(n,\mbf{k})$-neighbour is at most $K^n\cdot e^{-n^2/(10K)}$. For $0 < \delta \leq 1$, if $K\leq \delta \sqrt{n/\log(n)}$, then $K^2 \log(K) \leq \delta n$ and $n/K\geq K\log(K)/\delta\geq \log(K)/\delta$, so $e^{-n^2/(10K)}$ is extremely small. By Lemma~\ref{lem:union_bound}, the same is true of $K^n\cdot e^{-n^2/(10K)}$, which completes the proof of the lemma.
\end{proof}

It remains to give the (slightly more involved) proof of Lemma~\ref{lem:good_3path}.

\begin{proof}[Proof of Lemma~\ref{lem:good_3path}]
    In this proof we will relabel $V(n,\mbf{k})$ as $\prod_{i=1}^n\{0,\dots,k_i-1\}$ in the natural way and will consider these vertices as elements of the vector space $\mathbb{R}^n$. We will write $\mbf{e}_1,\dots,\mbf{e}_n$ for the standard basis of this space.

    We need to show that if $\delta>0$ is sufficiently small, then whenever $K\leq \delta \sqrt{n/\log(n)}$, the probability that there exists a good vertex that cannot reach some other good vertex within $H(n,\mbf{k})$-distance 3 of it is at most $e^{-cn\log(K)/\delta}$ for some universal $c>0$.
    Thus, let $\delta>0$ be small and assume that $K\leq \delta \sqrt{n/\log(n)}$. Note that since $n\geq 2$, by choosing $\delta$ small enough we may assume that $n/K^2 \geq \log(n)/\delta^2$ (and hence also $n$ and $n/K$) is large in absolute terms.

    Our proof will focus on pairs of vertices at distance exactly $3$ from one another, and it will be clear how to adapt the argument to pairs at distance $1$ or $2$.
    Let $u$ and $v$ be vertices of $H(n,\mbf{k})$ at distance $3$ from each other. After relabelling, we may assume that $u=\mbf{0}$ and $v=\mbf{e}_1+\mbf{e}_2+\mbf{e}_3$.
    Fix subsets $A',B'\subseteq [n]$ of sizes at least $n/4$ and $n/3K$ respectively; later we will assume that $u$ and $v$ are good vertices and take these to be the sets of coordinate directions in which $u$ and $v$ do not win and, respectively, win their lines. For each $i\in A'$ fix some $\alpha_i\in [k_i-1]$; later we will take $\alpha_i$ to be the $i$th coordinate of the winner of the line through $u$ in direction $i$. Now pick any $A\subseteq A'\setminus\{1,2,3\}$ and $B\subseteq B'\setminus\{1,2,3\}$ such that $\abs{A}=\ceil{n/5}$ and $\abs{B}=\ceil{n/4K}$. Relabelling again, we may assume that $A,B\subseteq[\floor{n/2}]$. 

    Having fixed $A$, $B$, and the $\alpha_i$, we will now define a certain type of path in $\vv{H}(n,\mbf{k})$. First, let $i,j\in [n]$ be distinct with $i\in A$ and $j\in B$, then let $\beta\in [k_j-1]$.
    A directed path in $\vv{H}(n,\mbf{k})$ from $\alpha_i\mbf{e}_i+\beta\mbf{e}_j$ to $\alpha_i\mbf{e}_i+\beta\mbf{e}_j + v$ will be called an \emph{$(i,j,\beta)$-path} if it has the following form: the path starts at $\alpha_i\mbf{e}_i+\beta\mbf{e}_j$ then follows a path of length 3 to $\alpha_i\mbf{e}_i+\beta\mbf{e}_j+\mbf{e}_1$ in which the first and third edges are in a coordinate direction taken from the interval of integers $[\floor{n/2}+1,\floor{2n/3}]$. That is, the path starts
    \[\alpha_i\mbf{e}_i+\beta\mbf{e}_j, \alpha_i\mbf{e}_i+\beta\mbf{e}_j + \gamma \mbf{e}_\ell, \alpha_i\mbf{e}_i+\beta\mbf{e}_j + \gamma \mbf{e}_\ell + \mbf{e}_1, \alpha_i\mbf{e}_i+\beta\mbf{e}_j + \mbf{e}_1,\]
    for some $\ell \in [\floor{n/2}+1,\floor{2n/3}]$ and $\gamma\in[k_\ell -1]$.
    Next, the path follows a path of length $3$ to $\alpha_i\mbf{e}_i+\beta\mbf{e}_j+\mbf{e}_1+\mbf{e}_2$ in which the first and third edges are in a coordinate direction taken from $[\floor{2n/3}+1,\floor{5n/6}]$, before finally following a path of length $3$ to $\alpha_i\mbf{e}_i+\beta\mbf{e}_j+\mbf{e}_1+\mbf{e}_2+\mbf{e}_3=\alpha_i\mbf{e}_i+\beta\mbf{e}_j+v$ in which the first and third edges are in a coordinate direction taken from $[\floor{5n/6}+1,n]$.

    Let $E_{(i,j,\beta)}^{(1)}$ be the event that there exists in $\Lnk$ a path of length $3$ from $\alpha_i\mbf{e}_i+\beta\mbf{e}_j$ to $\alpha_i\mbf{e}_i+\beta\mbf{e}_j+\mbf{e}_1$ in which the first and third edges are in a coordinate direction taken from $[\floor{n/2}+1,\floor{2n/3}]$. That is, $E_{(i,j,\beta)}^{(1)}$ is the event that there is some $\ell \in [\floor{n/2}+1,\floor{2n/3}]$ and some $\gamma \in [k_\ell - 1]$ such that all of the edges in the directed path
    \[
    \alpha_i\mbf{e}_i+\beta\mbf{e}_j, \alpha_i\mbf{e}_i+\beta\mbf{e}_j + \gamma \mbf{e}_\ell, \alpha_i\mbf{e}_i+\beta\mbf{e}_j + \gamma \mbf{e}_\ell + \mbf{e}_1, \alpha_i\mbf{e}_i+\beta\mbf{e}_j + \mbf{e}_1
    \]
    are present in $\Lnk$.
    Define $E_{(i,j,\beta)}^{(2)}$ and $E_{(i,j,\beta)}^{(3)}$ analogously for the second and third parts of the $(i,j,\beta)$-path. Note that there exists an $(i,j,\beta)$-path in $\Lnk$ if and only if all three of these events occur.

    The event that $\Lnk$ contains a path of length $3$ from $\alpha_i\mbf{e}_i+\beta\mbf{e}_j$ to $\alpha_i\mbf{e}_i+\beta\mbf{e}_j+\mbf{e}_1$ in which the first and third edges are in a given coordinate direction $\ell$ has probability
    \[
    \left(1-\frac{1}{k_\ell}\right)\frac{1}{k_1}\frac{1}{k_\ell}\geq \frac{1}{2K^2}.
    \]
    Indeed, the probability that $\alpha_i\mbf{e}_i+\beta\mbf{e}_j$ does not win its line in direction $\ell$ is $1-1/k_{\ell}$; 
    the probability that $\alpha_i\mbf{e}_i+\beta\mbf{e}_j+\mbf{e}_1$ wins its line in direction $\ell$ is $1/k_{\ell}$, the probability that the required edge in the direction 1 is present is $1/k_1$, and these three events occur independently.
    Since the existence of such a path is independent for different $\ell$, the failure probability of $E_{(i,j,\beta)}^{(1)}$ is at most
    \[
        \prod_{\ell=\floor{n/2}+1}^{\floor{2n/3}}{\left(1-\frac{1}{2K^2}\right)}
        \leq \left(1-\frac{1}{2K^2}\right)^{n/7}
        \leq e^{-n/(14K^2)}
        < \frac{1}{2},
    \]
    where we have used that $1+x\leq e^{x}$ for all $x\in \mathbb{R}$ and that $n/K^2$ is large.

    Similarly, $E_{(i,j,\beta)}^{(2)}$ and $E_{(i,j,\beta)}^{(3)}$ each occur with probability at least $1/2$. Moreover, it is not difficult to see that the sets of lines on whose presence each of these three events depend are pairwise disjoint, from which it follows that the events are independent. We deduce that $\Lnk$ contains an $(i,j,\beta)$-path with probability at least $1/8$.

    Next, we will call a path in $\vv{H}(n,\mbf{k})$ an \emph{extended $(i,j,\beta)$-path} if it is an $(i,j,\beta)$-path extended by one vertex at the end to $\beta\mbf{e}_j+v$. The line containing
    $\beta\mbf{e}_j+v$ and $\alpha_i\mbf{e}_i+\beta\mbf{e}_j+v$ could never be used in any $(i,j,\beta)$-path, so the probability that $\Lnk$ contains an extended $(i,j,\beta)$-path is at least $1/(8k_i)$.

    To conclude our definitions, a path in $\vv{H}(n,\mbf{k})$ will be called an \emph{$(i,j)$-path} if there exists $\beta\in[k_j-1]$ for which it is an extended $(i,j,\beta)$-path extended by one vertex at the start to $\alpha_i\mbf{e}_i$. Note that the line through $\alpha_i\mbf{e}_i$ in coordinate $j$ cannot be used in an extended $(i,j,\beta)$-path for any $\beta\in[k_j-1]$. Hence, the probability that $\Lnk$ contains an $(i,j)$-path is the probability that the winner of the line through $\alpha_i\mbf{e}_i$ in coordinate $j$ is some $\alpha_i\mbf{e}_i+\beta \mbf{e}_j$ with $\beta\in [k_j-1]$, multiplied by the probability that $\Lnk$ contains an extended $(i,j,\beta)$-path for this $\beta$. By the above, this probability is at least 
    \[
    \frac{1}{8k_i}\left(1-\frac{1}{k_j}\right)\geq \frac{1}{16k_i}.
    \]
    
    Next, observe that every line that could possibly be used in an $(i,j)$-path identifies the set $\{i,j\}$. There are at least $\abs{A}(\abs{B}-1)/2\geq n^2/(50K)$ ways to choose $\{i,j\}$, so the probability that $\Lnk$ does not contain an $(i,j)$-path for any $(i,j)$ is at most 
    \[
    \left(1-\frac{1}{16k_i}\right)^{n^2/(50K)}\leq e^{-n^2/(800K^2)}
    \]
    Observe also that the event that there exists an $(i,j)$-path in $\Lnk$ for some $(i,j)$ is independent of the behaviour of any lines of $V(n,\mbf{k})$ containing $u$ or $v$.

    This analysis holds for any choice of $A'$ and $B'$ containing at least $n/4$ and $n/3K$ directions respectively, so if $u$ and $v$ are good vertices, then we may take $A'$ to be the set of coordinate directions in which $u$ does not win its line, and $B'$ to be the set of directions in which $v$ does win its line. For each $i\in A'$, let $\alpha_i\in [k_i-1]$ be the $i$th coordinate of the winner of the line through $u$ in direction $i$. This means that the edges from $u$ to $\alpha_i \mbf{e}_i$ and from $\beta \mbf{e}_j + v$ to $v$ are both present for any choice of $i \in A$, $j \in B$, and $\beta \in[k_j-1]$, so if there is no path from $u$ to $v$ in $\Lnk$, then there is no $(i,j)$-path in $\Lnk$ for any $(i,j)$. We have shown that this happens with probability at most $e^{-n^2/(800K^2)}$.

    By a similar argument, the same holds for every pair of vertices at distance $1$ or~$2$ from each other. A union bound yields that the probability that there exists a good vertex which cannot reach in $\Lnk$ some other good vertex within $H(n,\mbf{k})$-distance 3 of it is at most $K^{2n}\cdot e^{-n^2/(800K^2)}$. Clearly $e^{-n^2/(800K^2)}$ is extremely small, so Lemma~\ref{lem:union_bound} implies that $K^{2n}\cdot e^{-n^2/(800K^2)}$ is also extremely small, which completes the proof of the lemma.
\end{proof}

\section{Proof of Theorem~\ref{thm:main_grids}: accessing good vertices}
\label{sec:main_proof}

Recall that in Section~\ref{sec:outline} and Appendix~\ref{sec:strong_cmpnt} we defined a vertex of $\Lnk$ to be \emph{good} if it wins at least $n/3K$ but at most $3n/4$ of its lines. We will start by considering the following event.

\begin{definition}[Event $C$]
    Let $C$ be the event that every vertex in $V(n,\mbf{k})$ that can be reached from more than $n/2$ vertices in $\Lnk$ can be reached from a good vertex.
\end{definition}

We show that $C$ is very likely.

\begin{lemma}\label{lem:event_C}
    The event $C$ occurs with extremely high probability.
\end{lemma}
\begin{proof}
    We need to show that if $\delta>0$ is sufficiently small, then whenever $K\leq \delta \sqrt{n/\log(n)}$, we have $\Prb(\ov{C})\leq e^{-c n\log(K)/\delta}$ for some $c >0$.
    To this end, let $\delta>0$ be small and assume that $K\leq \delta \sqrt{n/\log(n)}$. For each $y\in V(n,\mbf{k})$, define $C_y$ to be the event that either $y$ can be reached from at most $n/2$ vertices in $\Lnk$, or $y$ can be reached from a good vertex in $\Lnk$

    Let $M=\floor{n/3}$. Fix $y$ and let $S$ be the set of trees of order $M$ in $H(n,\mbf{k})$ that contain $y$. For a tree $T\in S$, let $C_T$ be the event that all the edges of $T$ are oriented towards $y$ in $\Lnk$ (so in particular, if~$C_T$ holds, then all vertices of $T$ can reach $y$). If $\ov{(C_y)}$ holds, then $y$ can be reached from more than $n/2\geq M$ vertices, so $C_T$ occurs for some $T\in S$. Hence, by a union bound
    \begin{equation}\label{eq:C}
        \Prb\big(\ov{(C_y)}\big)
        = \Prb\left(\ov{(C_y)} \cap \left(\bigcup_{T\in S}C_T\right)\right)
        \leq \sum_{T\in S}\Prb\big(\ov{(C_y)}\cap C_T\big).
    \end{equation}

    For fixed $T\in S$, if $\ov{(C_y)}$ and $C_T$ both hold, then no vertex in $T$ is good. Since $T$ contains exactly $M$ vertices, each of its vertices can be assigned a set of $n-M= \ceil{2n/3}$ lines which contain that vertex and no other vertex of~$T$ (so that the number of lines that each vertex of $T$ wins out of the $\ceil{2n/3}$ assigned to it is independent).
    For a fixed vertex $v$ of $T$, let $X$ be the number of the $\ceil{2n/3}$ lines assigned to $v$ in which~$v$ is the winner, so that $X$ is a sum of $\ceil{2n/3}$ independent Bernoulli random variables and has mean $\mu\geq 2n/3K$.

    By Lemma~\ref{lem:chernoff},
    \[
    \Prb\left(X\leq\frac{n}{3K}\right)
    \leq \Prb\left(X\leq\frac{\mu}{2}\right)
    \leq e^{-\mu/8}
    \leq e^{-n/(12K)}.
    \]
    Meanwhile, the number of the lines assigned to $v$ in which $v$ is not the winner, $\ceil{2n/3}-X$, is a sum of $\ceil{2n/3}$ independent Bernoulli random variables with mean $\mu'\geq n/3$. Hence, by Lemma~\ref{lem:chernoff} again,
    \[
    \Prb\Big(X> \frac{5n}{12}\Big)
    \leq \Prb\Big(\Big\lceil\frac{2n}{3}\Big\rceil-X\leq \frac{n}{4}\Big) 
    \leq \Prb\Big(\Big\lceil\frac{2n}{3}\Big\rceil-X\leq \frac{3\mu'}{4}\Big) 
    \leq e^{-\mu'/32}
    \leq e^{-n/96}.
    \]
    
    If $n/3K<X\leq 5n/12$, then in total $v$ wins at least $n/3K$ and at most $5n/12+\floor{n/3}\leq 3n/4$ of its lines, and hence is a good vertex.  
    It follows that each vertex of $T$ is good with failure probability at most $e^{-n/(50K)}$ (for $n \geq 101$), so at least one of these vertices is good with failure probability at most $e^{-Mn/(50K)}\leq e^{-n^2/(200K)}$ (by choosing $\delta$ small enough, we can ensure $M\geq n/4$).
    This is therefore an upper bound on $\Prb(\ov{(C_y)}\cap C_T)$. It follows from Lemma~\ref{lem:trees} that $|S|\leq (enK)^{M-1}$, so by~\eqref{eq:C} we have
    \begin{equation*}
        \begin{split}
            \Prb\big(\ov{(C_y)}\big)
            & \leq (enK)^{M-1}\cdot e^{-n^2/(200K)} \\
            & \leq \exp\left(\frac{n\log(enK)}{3}-\frac{n^2}{200K}\right) \\
            & = \exp\left(\frac{n^2}{K}\left(\frac{K\log(enK)}{3n}-\frac{1}{200}\right)\right).
        \end{split}
    \end{equation*}
    If $\delta$ is small enough, then $n$ is large relative to $K\log(enK)$, so this probability is at most $e^{-n^2/(300K)}$.

    Finally, since $C=\bigcap_y{C_y}$, by a union bound we have
    \[
        \Prb(\ov{C})\leq K^n\cdot e^{-n^2/(300K)}.
    \]
    Clearly $e^{-n^2/(300K)}$ is extremely small, so by Lemma~\ref{lem:union_bound} the same is true of $K^n\cdot e^{-n^2/(300K)}$, and the lemma follows.
\end{proof}

Next, we prove a very similar result for vertices that can \emph{reach} more than $n/2$ vertices rather than can be \emph{reached from} more than $n/2$ vertices.

\begin{definition}[Event $D$]
    Let $D$ be the event that every vertex in $\Lnk$ that can reach more than $n/2$ vertices can reach a good vertex.
\end{definition}

By an argument similar to that used to prove Lemma~\ref{lem:event_C}, we obtain the following.

\begin{lemma}\label{lem:event_D}
    The event $D$ occurs with extremely high probability.
\end{lemma}

We are now ready to put everything together to prove the main theorem.

\begin{proof}[Proof of Theorem~\ref{thm:main_grids}]
    Let $n\geq 2$ be an integer, let $\mbf{k}\in\{2,3,\dots\}^n$, and let $\eps>0$. Define $K=\max_i(k_i)$ and $N=(1+\eps)K\log(K)$. Let $E_\eps$ be the event that every vertex of $\Lnk$ can either be reached from at most $N$ vertices or can be reached from every non-sink; we want to show that $E_\eps$ occurs wvhp.

    Let events $A_\eps$, $B$, $C$, and $D$ be as above, and suppose that they all occur simultaneously in $\Lnk$. Let $x,y\in V(n,\mbf{k})$ where $x$ is a non-sink and $y$ can be reached from more than $N$ vertices. Since $A_\eps$ occurs, $x$ and $y$ can reach and be reached from more than $n/2$ vertices respectively. Thus, since $C$ occurs, there exists a good vertex $v\in V(n,\mbf{k})$ which can reach $y$ in $\Lnk$. Next, since $D$ occurs, there exists a good vertex $u\in V(n,\mbf{k})$ which can be reached from $x$ in $\Lnk$. Since $B$ occurs, all good vertices of $\Lnk$ are in the same strongly connected component, so in particular $u$ can reach $v$.
    It follows that there is a directed walk from $x$ to $y$ in $\Lnk$ via $u$ and $v$, that is, $x$ can reach $y$. In other words, if $A_\eps$, $B$, $C$, and $D$ occur, then so does $E_\eps$.
    By Lemmas~\ref{lem:good_cmpnt},~\ref{lem:event_C}, and~\ref{lem:event_D}, each of $B$, $C$, and $D$ occurs wehp, so in particular wvhp, and $A_\eps$ occurs wvhp by Lemma~\ref{lem:foothold}. Hence, by (repeated applications of) Lemma~\ref{lem:add_probs}, we conclude that $E_\eps$ occurs wvhp, as required.
\end{proof}

\section{On the tightness of Theorems~\ref{thm:main_game} and~\ref{thm:main_grids}}\label{sec:tech_tight}

In this section we discuss to what extent various aspects of Theorem~\ref{thm:main_game} and Theorem~\ref{thm:main_grids} are tight. First, with regards to the relationship between $n$ and $K$, it is entirely possible that this condition could be weakened considerably while still allowing results in the spirit of Theorem~\ref{thm:main_game} and Theorem~\ref{thm:main_grids}. However, the condition given in Theorem~\ref{thm:main_grids} seems to be at the limit of our methods, and a new approach would be needed to improve it. See also point (i) of our open questions in Section~\ref{sec:conc}.

Next, with regards to the tightness of the failure probability in the theorems, first note that the failure probability in Theorem~\ref{thm:main_grids} carries through to Theorem~\ref{thm:main_game} by the proof of Corollary~\ref{cor:conn}. The following theorem shows that this slightly stronger lower bound is tight up to the value of the exponent $c$. In fact, the theorem shows that even if Theorem~\ref{thm:main_game} were weakened to only consider weakly acyclic games rather than connected games, the probability inherited from Theorem~\ref{thm:main_grids} would still be tight up to the value of the constant in the exponent.

\begin{restatable}{theorem}{probBound}\label{thm:prob_bound}
    There is a constant $c' > 0$ such that
    \[
        \frac{| \{g\in\cG(n,\mbf{k})\colon g \text{\textnormal{ is weakly acyclic}}\} |}{| \{ g\in\cG(n,\mbf{k})\colon g \text{\textnormal{ is has a pure Nash equilibrium}}\} |} \leq 1- \prod_{i=1}^{n}{k_i^{-c'}}
    \]
    for all integers $n\geq 4$ and all $\mbf{k}\in\{2,3,\dots\}^n$.
\end{restatable}
\begin{proof}
    Let $n$ and $\mbf{k}$ be as in the statement of the theorem. Define a 4-cycle in a subgraph of $\vv{H}(n,\mbf{k})$ to be \emph{sticky} if each of its vertices can only reach the other vertices of the 4-cycle. The probability that a given sticky 4-cycle (where the edges in the cycle are the first and second coordinate directions, say) appears in $\Lnk$ is $k_1^{-2}k_2^{-2}\prod_{i=3}^n{k_i^{-4}}\geq\prod_{i=1}^n{k_i^{-4}}$. Since $n\geq 4$, there is a vertex none of whose lines intersect this sticky 4-cycle. The event that this vertex is a sink, which occurs with probability $\prod_{i=1}^n{k_i^{-1}}$, is therefore independent of whether or not the sticky 4-cycle appears, so with probability at least $\prod_{i=1}^n{k_i^{-5}}$ there is both a sink and a sticky 4-cycle in $\Lnk$.

    The vertices of a sticky 4-cycle cannot reach a sink, so the result now follows from arguments similar to those used to prove Corollary~\ref{cor:conn}.
\end{proof}

One might ask whether taking a larger value for $N$ in the statement of Theorem~\ref{thm:main_grids} would allow a significantly smaller failure probability, but a similar argument to the proof of Theorem~\ref{thm:prob_bound} shows that it would not. Indeed, there is some $c' > 0$ such that if $n$ is large relative to $\max_i(k_i)$, then with probability at least $1 - \prod_{i=1}^n{k_i^{-c'}}$ there is a vertex in $\Lnk$ which can be reached from $\prod_{i=1}^{n-1}{k_i}$ vertices but not from every non-sink.

This follows from an argument similar to the above: suppose that the desired sticky 4-cycle has vertices $(1,1,1,\dots,1)$, $(1,2,1,\dots,1)$, $(2,1,1,\dots,1)$, and $(2,2,1,\dots,1)$, then the subgraph $G$ of $\Lnk$ induced on $\prod_{i=1}^{n-1}[k_i] \times \{2\}$ has the same distribution as $\vv{L}(n-1,(k_1,\dots,k_{n-1}))$, and behaves independently of whether the desired sticky 4-cycle appears or not. Applying Theorem~\ref{thm:main_grids} to $G$ and using work of~\cite{rinott2000number}, one can show that there exists $p>0$ such that if $n$ is large enough relative to $\max_i(k_i)$, then with probability at least $p$, $G$ contains exactly one sink and this can be reached from every vertex in~$G$. It follows that there exists $c'>0$ such that if $n$ is large relative to $\max_i(k_i)$, then with probability at least $\prod_{i=1}^n{k_i^{-c'}}$ there is a vertex in $\Lnk$ which can be reached from $\prod_{i=1}^{n-1}{k_i}$ vertices but not from every non-sink.

Finally, to what extent is it possible to take a smaller $N$ in the statement of Theorem~\ref{thm:main_grids}?
It turns out that, for large $K$, the value of $N$ cannot be substantially improved as a function of $K$: it is possible to take $r$ not much smaller than $\log(K - 1)$ in the following theorem,\footnote{By applying the mean value theorem to $\log$ one can show that $\log(K)-\log(K-1) = 1/K + O(1/K^2)$.} so one cannot hope for a value of $N$ any better than $K \log(K) - O(\log(K))$. Here we let $\mbf{K} = (K, \dots, K)$ denote the all $K$'s vector of the appropriate length.

\begin{restatable}{theorem}{lowerBound}\label{thm:lower_bound}
    There is a constant $c > 0$ such that for all integers $n \geq 2$, $2\leq K \leq \sqrt{n}$, and
    \[
        1 \leq r \leq \frac{\log(K-1)}{(K - 1) (\log(K) - \log(K-1))},
    \]
    the probability that there is a vertex in ${\vv{L}(n, \mbf{K})}$ which can be reached from exactly $r(K-1) + 1$ vertices is at least $1 - c/n$.
\end{restatable}
\begin{proof}
    Let $f(K)$ denote the expression upper bounding $r$ in the theorem. One can show that this is increasing for $K\geq 2$ and that $f(\sqrt{n})\leq n$ for $n\geq 2$, so letting $n$, $K$, and $r$ be as in the statement, we have $r\leq n$. Let $X_a$ be the indicator random variable of the event that $a\in V(n,\mbf{K})$ wins exactly $r$ of its lines and every vertex on those $r$ lines except~$a$ is a source, and write $X = \sum_{a \in [K]^n} X_a$. We wish to upper bound the probability that $X=0$, for which we will use a second moment calculation.

    First, note that $X_a$ and $X_b$ are independent if the Hamming distance between $a$ and~$b$ (i.e.\ the number of coordinates on which $a$ and $b$ differ), denoted by $d(a,b)$, is at least four. It follows that
    \begin{align*}
        \mathbb{E}[X^2]
         & = \sum_{a \in [K]^n} \sum_{b \in [K]^n} \mathbb{P}(X_a X_b = 1)\\
         & \leq \sum_{a \in [K]^n} \sum_{b \in [K]^n} \mathbb{P}(X_a = 1) \mathbb{P}(X_b = 1) + \sum_{a \in [K]^n} \sum_{b \in [K]^n : d(a,b) \leq 3} \mathbb{P}(X_a=1)\\
         & \leq \mathbb{E}[X]^2 + K^3 n^3\, \mathbb{E}[X]\\
         & \leq \mathbb{E}[X]^2 + n^{9/2}\, \mathbb{E}[X].
    \end{align*}
    Hence, to apply Chebyshev's inequality, we need to show that $\mathbb{E}[X]$ grows more quickly than $n^{9/2}$. We have
    \begin{align*}
        \mathbb{E}[X]
         & = K^n \binom{n}{r} \frac{1}{K^r} \left( 1 - \frac{1}{K} \right)^{n -r + r(n-1)(K-1)}\\
         & \geq \left(\frac{n}{Kr}\right)^r \left[K\left( 1 - \frac{1}{K} \right)^{1 + (K-1)r}\right]^n\\
         & \geq  \left(\frac{\sqrt{n}}{r}\right)^r \left[K\left( 1 - \frac{1}{K} \right)^{1 + (K-1)r}\right]^n,
    \end{align*}
    where we have used $\binom{n}{r}\geq (n/r)^r$ in the second line and $K\leq \sqrt{n}$ in the last line.

    We will analyse the two terms in this product separately. For fixed $n$, the first term, $(\sqrt{n}/r)^r$, is increasing for $r\in[0,\sqrt{n}/e]$. Since $f(\sqrt{n})\leq \sqrt{n}/e$ for all $n\geq 2$, it follows that this term is always at least $\sqrt{n}\geq 1$, and if $r\geq 11$, then it is at least $n^{11/2}/{11}^{11}$.
    
    Turning to the second term, it is straightforward to check that if $f(K)=r$, then the expression in square brackets is equal to $1$, and that for fixed $K\geq 2$ this expression is strictly decreasing for $r\in [1,f(K)]$. It follows that the second term is always at least~$1$. There are no integer solutions $K$ to $f(K)=r$ for any $r\in[10]$, so for $r\leq 10$ the expression in square brackets is always strictly greater than $1$. Moreover, for fixed~$r$ we have $K(1-1/K)^{1+r(K-1)}\to \infty$ as $K\to \infty$, so in fact there exists some universal $\eps>0$ such that the second term is at least $(1 + \eps)^n$ whenever $r \leq 10$.
    
    Combining, we have $\mathbb{E}[X]\geq \min\big\{(1 + \eps)^n, n^{11/2}/{11}^{11}\big\}$ for all admissible $n$, $K$, and~$r$. By Chebyshev's inequality, this yields
    \[
        \mathbb{P}\left(|X - \mathbb{E}[X]| \geq \tfrac{1}{2}\mathbb{E}[X] \right)
        \leq 4\cdot \frac{\mathbb{E}[X^2] - \mathbb{E}[X]^2}{\mathbb{E}[X]^2}
        \leq \frac{4n^{9/2}}{\min\left\{ (1 + \eps)^n, n^{11/2}/{11}^{11} \right\}},
    \]
    and since $\mathbb{E}[X]>0$ this is, in turn, an upper bound on $\Prb(X=0)$.
    There is a constant $c'>0$ such that this upper bound is at most $c'/n$ for all sufficiently large $n$, and we can choose $c>0$ large enough to accommodate the (finitely many) remaining cases.
\end{proof}

If, however, we are prepared to allow the exponent in the failure probability to depend on $K$, then we can adapt the proof of Theorem~\ref{thm:main_grids} to slightly improve the value of $N$.

\begin{restatable}{theorem}{constantK}\label{thm:constant_K}
    For all integers $K\geq 2$, there exists $c_K>0$ such that for all integers $n\geq 2$ and all $\mbf{k}\in\{2,3,\dots K\}^n$, every vertex of $\Lnk$ can either be reached from at most
    \[
        N'=\frac{\log(K)}{\log(K)-\log(K-1)}
    \]
    vertices or from every non-sink, with failure probability at most $e^{-c_K n}$.
\end{restatable}
\begin{proof}
    Note first that it is sufficient to show that the result holds when $n$ is large relative to $K$, since this covers all but finitely many cases for each $K$, and $c_K$ can be chosen to handle these.

    By the $\eps=1$ case of Theorem~\ref{thm:main_grids}, there exists $c>0$ such that if $n$ is large relative to~$K$, then for all $\mbf{k}\in\{2,\dots,K\}^n$, the failure probability of the event that every vertex of $\Lnk$ can either be reached from at most $2K\log(K)$ vertices, or from every non-sink, is at most $e^{-c n}$.
    Hence, to prove the theorem it is enough to show that for all $K$ there exists $c'_K>0$ such that if $n$ is large enough relative to $K$, then for all $\mbf{k}\in\{2,\dots,K\}^n$, with failure probability at most $e^{-c'_K n}$, no vertices of $\Lnk$ can be reached from more than $N'$ vertices but at most $2K\log(K)$ vertices.

    This can be achieved by modifying the proof of Lemma~\ref{lem:Y_small}. Defining $Y_m$ as in Appendix~\ref{sec:foothold}, we need to show that $Y_m=0$ for all $N'<m\leq 2K\log(K)$ with failure probability $e^{-c'_K n}$. Fix such an $m$, then as in the proof of Lemma~\ref{lem:Y_small} we have
    \[
        \mathbb{E}[Y_m]
        \leq \left[4nK^2 \cdot \left(K^{1/m}\left(1-\frac{1}{K}\right)\right)^{n-m}\right]^m.
    \]
    In place of~\eqref{eq:K_constant_different}, it is not difficult to check that $m>N'$ ensures that $K^{1/m}(1-1/K)<1-\eta_K$ for some $\eta_K\in(0,1)$. Thus, for $n$ large in terms of $K$ (uniformly in $m$), we have
    \[
        \mathbb{E}[Y_m]
        \leq \left(4nK^2(1-\eta_K)^{n-m}\right)^m
        \leq \left(1-\eta_K^2\right)^{m(n-m)}
        \leq e^{-c''_K n},
    \]
    for some $c''_K>0$.
    Then, if $n$ is large relative to $K$,
    \[
        \sum_{N<m\leq 2K\log(K)}\mathbb{E}[Y_m]\leq 2K\log(K)\cdot e^{-c''_K n}\leq e^{-c'_K n},
    \]
    for some $c'_K>0$, so by Lemma~\ref{lem:markov} we have that $Y_m=0$ for all $N'<m\leq 2K\log(K)$ with failure probability at most $e^{-c'_K n}$, as required.
\end{proof}

Together, Theorem~\ref{thm:lower_bound} and Theorem~\ref{thm:constant_K} essentially determine the `correct' value for $N$ as a function of $K$. Indeed, if we ignore the fact that $r$ must be an integer in Theorem~\ref{thm:lower_bound}, then that theorem implies that when $n$ is much larger than $K$, the graph $\vv{L}(n,\mbf{K})$ typically contains a vertex that can be reached from `exactly' $\log(K)/(\log(K)-\log(K-1))$ vertices. Meanwhile Theorem~\ref{thm:constant_K} implies that typically every vertex which can be reached from more than this many vertices can be reached from every non-sink.

Although the improvement to the value of $N$ represented by Theorem~\ref{thm:constant_K} is modest, it has the following consequence for $\vv{L}(n,\mbf{2})$ and $\vv{L}(n,\mbf{3})$ which is of independent interest.

\begin{corollary}\label{cor:constant_K}
    There exists a constant $c > 0$ such that with failure probability at most $e^{-c n}$, every non-sink in $\vv{L}(n,\mbf{2})$ can reach every non-source. If $n\geq 2$, the same is true for $\vv{L}(n,\mbf{3})$. Conversely,
    for each $K \geq 4$, the probability that there is a non-sink in $\vv{L}(n,\mbf{K})$ which cannot reach every non-source tends to 1 as $n\to\infty$.
\end{corollary}
Corollary \ref{cor:constant_K} is restated as Proposition \ref{prop:234} in the main text.

The positive direction of the corollary follows straightforwardly from Theorem~\ref{thm:constant_K} and the observation that every non-source in $\Lnk$ can be reached from at least $\min_i(k_i)$ vertices, and the negative direction follows immediately from setting $r = 1$ in Theorem~\ref{thm:lower_bound}. While the $\mbf{k}=\mbf{2}$ case of the corollary follows from Theorem~\ref{thm:main_grids}, the $\mbf{k}=\mbf{3}$ case does not.

\section{Acyclicity of directed grids}
\label{sec:acyclic_grids_proof}\label{sec:tech_acyclic}

Using the terminology and notation of Section~\ref{sec:tech_conn} we state the following strengthening of Proposition~\ref{prop:acyclic_games}.
\begin{restatable}{proposition}{acyclicGrids}\label{prop:acyclic_grids}
    There exists $c>0$ such that for all integers $n\geq 2$ and all $\mbf{k}\in\{2,3,\dots\}^n$, the probability that $\Lnk$ is acyclic is at most $\exp(-cnk^{n-2})$, where $k\coloneqq \min_i(k_i)$.
\end{restatable}
The proof below actually yields a slightly better upper bound on the probability that $\Lnk$ is acyclic, but for clarity we have not included this in the statement.

For distinct $i,j\in[n]$, we define an \emph{$\{i,j\}$-plane} of $V(n,\mbf{k})$ to be a subset of $V(n,\mbf{k})$ of size $k_ik_j$ whose elements pairwise differ in at most their $i$th and $j$th coordinates. A subset of $V(n,\mbf{k})$ will be called a \emph{plane} of $V(n,\mbf{k})$ if it is an $\{i,j\}$-plane for some $i$ and $j$.

\begin{proof}[Proof of Proposition~\ref{prop:acyclic_grids}]
    We begin with the following claim.
    \begin{claim}
        Let $k_1,k_2\geq 2$ be integers, then $\vv{L}(2,(k_1,k_2))$ contains a cycle with probability at least $1/8$.
    \end{claim}
    \begin{proof}
        We will define a random process $\mbf{X}=(X_0,X_1,X_2,\dots)$ coupled to $\vv{L}(2,(k_1,k_2))$. Let $X_0=(1,1)$ and for each $t\geq 1$, given $X_{t-1}\in [k_1]\times[k_2]$, if $t$ is odd, let $X_t$ be the winner of the line in coordinate~$1$ which contains~$X_{t-1}$. If $t$ is even, let $X_t$ be the winner of the line in coordinate $2$ which contains~$X_{t-1}$. Thus, $\mbf{X}$ is a random walk on $[k_1]\times[k_2]$ starting at $(1,1)$, which at odd time steps traverses the available edge of $\vv{L}(2,(k_1,k_2))$ in the first coordinate direction (if there is one), and at even time steps traverses the available edge in the second.

        Let $T$ be the least time $t$ at which there exists $1\leq i<t$ such that $X_i$ and $X_t$ have the same first coordinate if $t$ is odd, or the same second coordinate if $t$ is even. It is not difficult to check (bearing in mind that we do not explore from $X_0$ in the second dimension at the start of the process) that for each $t<T$, when we explore from $X_t$ we do so in an unexplored line.

        Once we reach $X_T$, we already know the winner in the line we want to explore: it is the $X_i$ with $1\leq i<T$ that has the same first or second coordinate (depending on the parity of $T$) as $X_T$. Hence, $X_{T+1}=X_i$ for this $i$ and the process is deterministic from here, with $X_{T+2}=X_{i+1}$ and so on. The process either becomes stationary at this point or enters a (non-trivial) cycle. It is straightforward to see that the process is stationary exactly when $X_T=X_{T-1}$, so if $X_T\neq X_{T-1}$, then $\vv{L}(2,(k_1,k_2))$ contains a cycle. Hence, let $A$ be the event that $X_T\neq X_{T-1}$. We will show that $\Prb(A)\geq 1/8$.

        Given $T$ and $X_{T-1}$, $X_T$ is chosen uniformly at random from among those vertices in the unexplored line through $X_{T-1}$ that also belong to a previously explored line. There are $\floor{T/2}$ such vertices to choose from. Writing $\tau=\tau(k_1,k_2)$ for the maximum possible value of $T$, it follows that, for each $t\in\{2,\dots,\tau\}$, we have $\Prb(A\,|\,T=t)=(\floor{t/2}-1)/\floor{t/2}$. This is at least $1/2$ for $t\geq 4$, so
        \[
            \Prb(A)=\sum_{t=2}^{\tau}{\Prb(A\,|\,T=t)\Prb(T=t)\geq \frac{\Prb(T\geq 4)}{2}.}
        \]
        We have $T\geq 4$ exactly when $X_1$ does not win its line in the second dimension and $X_2$ does not win its line in the first dimension. This event occurs with probability $(1-1/k_1)(1-1/k_2)\geq 1/4$, so $\Prb(A)\geq 1/8$, as required.
    \end{proof}
    Let $n\geq 2$ be an integer and let $\mbf{k}\in\{2,3,\dots\}^n$.
    By the claim, any given plane of $V(n,\mbf{k})$ induces a cyclic subgraph of $\Lnk$ with probability at least~$1/8$.
    In a family of planes which pairwise intersect in at most one vertex, each plane induces a cyclic subgraph of $\Lnk$ independently.
    The collection consisting of all $\{1,2\}$-planes, all $\{3,4\}$-planes, and so on, up to the $\{2\floor{n/2}-1, 2\floor{n/2}\}$-planes, is such a family. For distinct $i,j\in[n]$, the number of $\{i,j\}$-planes in $V(n,\mbf{k})$ is $\prod_{a\in[n]\setminus\{i,j\}}{k_a}$, so this family has size at least $\floor{n/2}\min(k_i)^{n-2}$, and the proposition follows.
\end{proof}

Proposition~\ref{prop:acyclic_games} follows from Proposition~\ref{prop:acyclic_grids} by the same reasoning with which we deduced Theorem~\ref{thm:main_game} from Theorem~\ref{thm:main_grids}: since $\Lnk$ has the same distribution as the best-response graph of a game drawn uniformly at random from all games in $\cG(n,\mbf{k})$, we have that $\Prb\big(\Lnk \text{ is acyclic}\big) / \Prb(S_{n,\mbf{k}})$ is equal to
\[
\frac{| \{g\in \cG(n,\mbf{k})\colon g \text{ is acyclic}\} |}{| \{ g\in\cG(n,\mbf{k})\colon g \text{ has a pure Nash equilibrium}\} |},
\]
where $S_{n,\mbf{k}}$ is the event that $\Lnk$ contains a sink, as in Section~\ref{sec:tech_conn}.
Thus, since $k_i\geq 2$ for all $i$, Proposition~\ref{prop:acyclic_games} follows from Proposition~\ref{prop:acyclic_grids} and the fact that $\Prb(S_{n,\mbf{k}})$ is at least a positive constant for all $n$ and $\mbf{k}$ under consideration, as noted in Section~\ref{sec:tech_conn}.

\section{Simulations}\label{sec:simulations}

\begin{figure*}
\centering
\begin{tikzpicture}

\begin{groupplot}[
        group style={
            group name=my plots,
            group size=1 by 3,
            vertical sep=1.5cm,
        },
        width=0.75\textwidth,
        height=5cm,
        xlabel={Number of agents},
        ylabel={Proportion},
        xmin=1.5, xmax=15.5,
        ymin=-0.10, ymax=1.10,
        xtick={2, 3, 4, 5, 6, 7, 8, 9, 10, 11, 12, 13, 14, 15},
        ytick={0,0.2,0.4,0.6,0.8,1},
        grid style=dashed,
        axis x line = bottom,
        axis y line = left,
        scale only axis,
    ]
    
\nextgroupplot[ylabel={Proportion connected}, legend to name=leg, legend columns=-1, legend style={draw=none, /tikz/every even column/.append style={column sep=2ex}}]

\addplot[
    color=color2,
    mark=*,
    mark size = 4pt,
    thick,
    fill opacity = 0.5,
    draw opacity = 1,
    ]
    coordinates {
    (2,1)
    (3,0.7951)
    (4, 0.7436)
    (5, 0.797)
    (6, 0.8550)
    (7, 0.9013)
    (8, 0.9477)
    (9, 0.9673)
    (10, 0.9796)
    (11, 0.9882)
    (12, 0.9946)
    (13, 0.9978) 
    (14, 0.9985)
    (15, 0.9994)
    };

\addplot[
    color=color3,
    mark=triangle*,
    thick,
    fill opacity = 0.5,
    draw opacity = 1,
    mark size = 4pt,
    ]
    coordinates {
    (2, 0.7068)
    (3, 0.6463)
    (4, 0.7852)
    (5, 0.8924)
    (6, 0.9541)
    (7, 0.9794)
    (8, 0.9917)
    (9, 0.9977)
    (10, 0.9989)
    (11, 0.9995)
    (12, 0.9999)
    (13, 1)
    (14, 1)
    (15, 1)
    };

\addplot[
    color=color4,
    mark=square*,
    thick,
    fill opacity = 0.5,
    draw opacity = 1,
    mark size = 4pt,
    ]
    coordinates {
    (2, 0.58531)
    (3, 0.6522)
    (4, 0.8498)
    (5, 0.9464)
    (6, 0.9858)
    (7, 0.9940)
    (8, 0.9986)
    (9, 0.9993)
    (10, 0.9998)
    (11, 1)
    (12, 0.9999)
    };

\addplot[
    color=color5,
    mark=pentagon*,
    thick,
    fill opacity = 0.5,
    draw opacity = 1,
    mark size = 4pt,
    ]
    coordinates {
    (2, 0.5078)
    (3, 0.6739)
    (4, 0.8871)
    (5, 0.9655)
    (6, 0.9936)
    (7, 0.9972)
    (8, 0.9996)
    (9, 0.9999)
    (10, 1)
    (11, 1)
    };

    \addlegendentry{2 actions};    
    \addlegendentry{3 actions};    
    \addlegendentry{4 actions};
    \addlegendentry{5 actions};    
    
\nextgroupplot[ylabel={Proportion acyclic}]

\addplot[
    color=color2,
    mark=*,
    mark size = 4pt,
    thick,
    fill opacity = 0.5,
    draw opacity = 1,
    ]
    coordinates {
    (2,1)
    (3,0.5941)
    (4, 0.0665)
    (5, 0)
    (6, 0)
    (7, 0)
    (8, 0)
    (9, 0)
    (10, 0)
    (11, 0)
    (12, 0)
    (13, 0)
    (14, 0)
    (15, 0) 
    };

\addplot[
    color=color3,
    mark=triangle*,
    thick,
    fill opacity = 0.5,
    draw opacity = 1,
    mark size = 4pt,
    ]
    coordinates {
    (2, 0.9686)
    (3, 0.1175)
    (4, 0)
    (5, 0)
    (6, 0)
    (7, 0)
    (8, 0)
    (9, 0)
    (10, 0)
    (11, 0)
    (12, 0)
    (13, 0)
    (14, 0)
    (15, 0)
    };

\addplot[
    color=color4,
    mark=square*,
    thick,
    fill opacity = 0.5,
    draw opacity = 1,
    mark size = 4pt,
    ]
    coordinates {
    (2, 0.8996)
    (3, 0.0054)
    (4, 0)
    (5, 0)
    (6, 0)
    (7, 0)
    (8, 0)
    (9, 0)
    (10, 0)
    (11, 0)
    (12, 0)
    };

\addplot[
    color=color5,
    mark=pentagon*,
    thick,
    fill opacity = 0.5,
    draw opacity = 1,
    mark size = 4pt,
    ]
    coordinates {
    (2, 0.8378)
    (3, 0.0002)
    (4, 0)
    (5, 0)
    (6, 0)
    (7, 0)
    (8, 0)
    (9, 0)
    (10, 0)
    (11, 0)
    };
    
\nextgroupplot[ylabel={Proportion super-connected}]

\addplot[
    color=color2,
    mark=*,
    mark size = 4pt,
    thick,
    fill opacity = 0.5,
    draw opacity = 1,
    ]
    coordinates {
    (2,0.1429)
    (3,0.1225)
    (4, 0.2312)
    (5, 0.4165)
    (6, 0.6056)
    (7, 0.7408)
    (8, 0.8542)
    (9, 0.9147)
    (10, 0.9525)
    (11, 0.9714)
    (12, 0.9858)
    (13, 0.9931)
    (14, 0.9947)
    (15, 0.9986) 
    };

\addplot[
    color=color3,
    mark=triangle*,
    thick,
    fill opacity = 0.5,
    draw opacity = 1,
    mark size = 4pt,
    ]
    coordinates {
    (2, 0)
    (3, 0.0166)
    (4, 0.0302)
    (5, 0.0334)
    (6, 0.0316)
    (7, 0.0279)
    (8, 0.0297)
    (9, 0.0304)
    (10, 0.0317)
    (11, 0.0350)
    (12, 0.0367)
    (13, 0.0407)
    (14, 0.0472)
    (15, 0.0587)
    };

\addplot[
    color=color4,
    mark=square*,
    thick,
    fill opacity = 0.5,
    draw opacity = 1,
    mark size = 4pt,
    ]
    coordinates {
    (2, 0)
    (3, 0.0013)
    (4, 0.0001)
    (5, 0)
    (6, 0)
    (7, 0)
    (8, 0)
    (9, 0)
    (10, 0)
    (11, 0)
    (12, 0)
    };

\addplot[
    color=color5,
    mark=pentagon*,
    thick,
    fill opacity = 0.5,
    draw opacity = 1,
    mark size = 4pt,
    ]
    coordinates {
    (2, 0)
    (3, 0)
    (4, 0)
    (5, 0)
    (6, 0)
    (7, 0)
    (8, 0)
    (9, 0)
    (10, 0)
    (11, 0)
    };

\end{groupplot}

\node[below] at (current bounding box.south) {\hspace{1.5cm}\pgfplotslegendfromname{leg}};

\end{tikzpicture}
\\
\emph{Note}. Error bars are omitted from the plots because they are narrow.
\caption{\emph{Simulations}. We drew ten thousand games uniformly at random from the set $\{ g \in \cG( n,  (k,\dots,k) ) : g \text{ has a pure Nash equilibrium} \}$ for various combinations of $n$ and $k$. The fraction of the drawn games that are connected, acyclic, and super-connected is shown on the vertical axes. }
\label{fig:simulations}
\end{figure*}

We have investigated the relative sizes of the game classes defined in Section~\ref{sec:games} for generic games with few players and few actions per player via simulation. Figure~\ref{fig:simulations} presents the simulation outcomes. We focused on games in which each of $n$ players has the same number of actions, $k$, and looked at various combinations of $n$ and $k$. For each combination, we drew ten thousand games uniformly at random from among all generic games with at least one pure Nash equilibrium, and we recorded the fraction of the drawn games that were connected, acyclic, or super-connected. As can be seen in Figure~\ref{fig:simulations}, the fraction of connected games gets close to 1 even for relatively small values of $n$ and $k$, and the fraction of acyclic games gets close to 0 even for relatively small values of $n$ and $k$. The fraction of super-connected games gets close to 1 as $n$ increases when $k=2$ and it is close to 0 for all $n$ when $k \geq 4$. For $k=3$, the fraction hovers above (but is close to) 0 for the small values of $n$ that we computed; based on our theoretical analysis, we expect it to eventually get close to 1 from about $n > 30$, but simulations rapidly become computationally expensive.

\end{document}